\documentclass[preprint,11pt]{elsarticle}

\usepackage{amssymb}
\usepackage{ifthen}
\usepackage[utf8]{inputenc}
\usepackage{latexsym, dsfont}
\usepackage{amsmath,amsthm}
\usepackage[inline]{enumitem}
\usepackage{comment}
\usepackage{siunitx}
\usepackage{graphicx}
\usepackage{caption}
\usepackage{hyperref}   
\usepackage[capitalize]{cleveref}
\usepackage{todonotes}
\usepackage[caption=false,font=footnotesize]{subfig}
\usepackage{mathtools}
\usepackage{tikz}
\newboolean{PGFPlots}
\setboolean{PGFPlots}{true}
\newcommand{\IR}{\mathds{R}}
\newtheorem{thm}{Theorem}
\newtheorem{lem}[thm]{Lemma}
\newtheorem{dfn}{Definition}

\DeclareMathOperator{\sgn}{sgn}

\usetikzlibrary{arrows,shapes,decorations,decorations.pathreplacing,decorations.pathmorphing,intersections,calc,fit}

\ifthenelse{\boolean{PGFPlots}}{
\usepackage{pgfplots}
\usepgfplotslibrary{units,groupplots}
\usetikzlibrary{external}
\tikzexternalize
\tikzset{external/only named=true}
}{}

\usepackage[appendix=append]{apxproof} 
\usepackage[append]{optional}

\newtheoremrep{theorem}{Theorem}[section]

\newcommand{\dup}{\delta_\uparrow}
\newcommand{\ddo}{\delta_\downarrow}
\newcommand{\fup}{f_\uparrow}
\newcommand{\fdo}{f_\downarrow}
\newcommand{\Fup}{F_\uparrow}
\newcommand{\Fdo}{F_\downarrow}
\newcommand{\dupD}{\delta^\uparrow}
\newcommand{\ddoD}{\delta^\downarrow}

\newcommand{\dsu}{\delta_S^\uparrow}
\newcommand{\dsd}{\delta_S^\downarrow}

\newcommand{\figPath}[1]{figures/#1}

\newcommand{\NOR}{\texttt{NOR}}
\newcommand{\NAND}{\texttt{NAND}}
\newcommand{\GOR}{\texttt{OR}}
\newcommand{\Gand}{\texttt{AND}}

\newcommand{\dmin}{\delta_{\mathrm{min}}}

\newcommand{\vth}{V_{th}}

\newcommand{\vdd}{V_{DD}}
\newcommand{\gnd}{\textit{GND}}
\newcommand{\vout}{V_{out}}

\newcommand{\dd}{\mathrm{d}}

\newcommand{\va}{V_{A}}
\newcommand{\vb}{V_{B}}

\newcommand{\vint}{V_{int}}
\newcommand{\cint}{C_{int}}

\newcommand{\cout}{C}

\newcommand{\nmos}{\texttt{nMOS}}
\newcommand{\pmos}{\texttt{pMOS}}

\newcommand{\ohm}{(OHM)}

\newcommand{\on}{\mbox{\emph{on}}}
\newcommand{\off}{\mbox{\emph{off}}}

\def\C{{\mathcal C}}
\newcommand{\IN}{\mathds{N}}
\usetikzlibrary{petri}
\tikzstyle{binary place}=[place,circle, double]
\tikzstyle{node}=[circle,draw=black,thick,minimum size=9mm]
\tikzstyle{dest}=[circle,draw=black!50,fill=black!20,thick,minimum
  size=9mm]
\tikzstyle{post}=[->,thick]
\tikzstyle{pre}=[<-,thick]
\tikzstyle{every transition}=[fill,minimum width=1cm,minimum height=2mm]
\tikzstyle{Atransition}=[transition,fill,minimum width=1cm,minimum height=2mm]
\tikzstyle{Otransition}=[transition,fill=white,minimum width=1cm,minimum height=2mm]
\tikzstyle{THtransition}=[transition,fill=white,minimum width=4mm,minimum height=1cm]

\tikzstyle{Tdelay} = [draw, rectangle, rounded corners,
    minimum height=4mm, minimum width=20mm]

\tikzstyle{Tfunction} = [draw, rectangle,
    minimum height=4mm, minimum width=4mm]

\tikzstyle{Tsignal} = [draw,fill=black,circle, size=1mm]

\tikzstyle{ra} = [draw,thick,double,double distance=1.0pt,->]
\tikzstyle{r} = [draw,->,line width=0.5pt]

\def\figanfup[#1]{(1-exp(-(#1)/0.15))*(1-exp(-(#1)/0.16))}
\def\figanfdown[#1]{(1-(1-exp(-((#1)^2)/0.01))*(1-exp(-((#1)^2)/0.1)))}

\def\figurangle{300}
\def\figurlen{1.2}
\def\figurdist{4.5}

\def\tikzfigureunrolling{
\begin{tikzpicture}[
gate/.style={draw, fill, circle, minimum size=3pt, inner sep=0},
xscale=1,
yscale=0.6,
>=latex'
]

\path
(0,0) node (I) [gate] {}
+(0:\figurlen) node (A) [gate] {}
++(\figurangle:\figurlen) node (B) [gate] {}
++(0:\figurlen) node (C) [gate] {}
++(0:\figurlen) node (O) [gate] {};

\foreach \from/\to in {I/A, A/C, I/B, B/C, C/O} {
\path [draw, ->] (\from) -- (\to);
}
\path [draw, <-, scale=0.8]
(B) .. controls +(-1,-1) and +(1,-1) .. (B);

\foreach \what/\anchor in {I/south, A/south, B/south, C/north, O/north} {
\node [anchor=\anchor] at (\what) {\footnotesize $\what$};
}

\pgftransformshift{\pgfpoint{+0.5cm}{0cm}} 

\path
(\figurdist,0) node (i1) [gate] {}
++(0:\figurlen) node (x1) [gate] {}
++(\figurangle:\figurlen) node (c1) [gate] {}
++(0:\figurlen) node (o1) [gate] {}
(c1) ++(0:-\figurlen) node (y1) [gate] {}
++(0:-\figurlen) node (y2) [gate] {}
+(\figurangle:-\figurlen) node (ot2) [gate] {};

\foreach \from/\to in {ot2/y2, y2/y1,
                       y1/c1, i1/y1, i1/x1, x1/c1, c1/o1, i1/y2} {
\path [draw, ->] (\from) -- (\to);
}

\foreach \where/\anchor/\text in
{ot2/south/$X_0$,
y2/north/$B^{(1)}$, i1/south/$I$, x1/south/$A^{(2)}$,
y1/north/$B^{(2)}$, o1/north/$O^{(3)}$, c1/north/$C^{(3)}$} {
\node [anchor=\anchor] at (\where) {\footnotesize \text};
}
\end{tikzpicture}
}

\def\figgawidth{2.8}
\def\figgaheight{1.7}
\def\figgayshift{0.15}

\def\figgamuxwidth{0.45}
\def\figgamuxheight{1}
\def\figgamuxindist{50}
\def\figgamuxtoout{0.25}
\def\figgaoutlen{0.25}
\def\figgarstlen{0.15}
\def\figgainlen{0.25}

\def\tikzfiguregate{
\begin{tikzpicture}[
gate/.style={draw, inner sep=0},
xscale=1,
yscale=0.4,
>=latex'
]

\coordinate (origin) at (0,0);
\coordinate (topright) at (\figgawidth,\figgaheight);

\path [draw] (origin) rectangle (topright);

\path [draw]
($(origin -| topright)!0.5!(topright) + (0,\figgayshift)$)
+(\figgaoutlen,0) node [anchor=west] {\small $v$} --
+(-\figgamuxtoout,0) node (mux) [draw, trapezium, anchor=top side,
shape border rotate=270, minimum width=\figgamuxheight*1cm, 
minimum height=\figgamuxwidth*1cm, trapezium stretches body] {};

\foreach \muxin/\muxincpt/\name/\cpt in {1/0/bool/$b$, -1/1/const/$I$} {
\path [draw] 
(mux.{180+\figgamuxindist*\muxin}) node [anchor=west, xshift=-1.5pt] {\footnotesize \muxincpt} --
++(-\csname figga\name tomux\endcsname, 0)
node (\name) [gate, anchor=east, minimum width=\csname figga\name width\endcsname*1cm,
minimum height=\csname figga\name height\endcsname*1cm] {\small \cpt};
}

\path [draw, densely dashed, thin]
(mux.north) -- (mux.north |- topright) --
++(0, \figgarstlen) node [anchor=south, inner sep=1.5pt] {\small\em reset};

\foreach \pos in {0.125, 0.375, 0.625, 0.875} {
\path [draw]
($(bool.north west)!\pos!(bool.south west)$) coordinate (tmp) --
(tmp -| origin) -- ++(-\figgainlen,0);
}
\end{tikzpicture}
}

\newboolean{conference}
\setboolean{conference}{true}  

\crefname{equation}{}{}

\begin{document}

\begin{frontmatter}

\title{Faithful Dynamic Timing Analysis of Digital Circuits Using Continuous Thresholded Mode-Switched ODEs}

\author{Arman Ferdowsi\fnref{label1}}
\ead{aferdowsi@ecs.tuwien.ac.at}
\author{Matthias Függer\fnref{label2}}
\ead{mfuegger@lmf.cnrs.fr}
\author{Thomas Nowak\fnref{label2,label3}}
\ead{thomas@thomasnowak.net}
\author{Ulrich Schmid\fnref{label1}}
\ead{s@ecs.tuwien.ac.at}
\author{Michael Drmota\fnref{label4}}
\ead{michael.drmota@tuwien.ac.at}

\fntext[label1]{TU Wien, Embedded Computing Systems Group} 
\fntext[label2]{LMF, Université Paris-Saclay, CNRS, ENS Paris-Saclay}
\fntext[label3]{Institut Universitaire de France}
\fntext[label4]{TU Wien, Institute of Discrete Mathematics and Geometry}

\begin{abstract}
Thresholded hybrid systems are restricted dynamical systems, where the current mode, and hence the ODE system describing its behavior, is solely determined by externally supplied digital input signals and where the only output signals are digital ones generated by comparing an internal state variable to a threshold value. An attractive feature of such systems is easy composition, which is facilitated by their purely digital interface. A particularly promising application domain of thresholded hybrid systems is digital integrated circuits: Modern digital circuit design considers them as a composition of Millions and even Billions of elementary logic gates, like inverters, \GOR\ and \Gand. Since every such logic gate is eventually implemented as an electronic circuit, however, which exhibits a behavior that is governed by some ODE system, thresholded hybrid systems are ideally suited for making the transition from the analog to the digital world
rigorous.

In this paper, we prove that the mapping from digital input signals to digital output signals is continuous for a large class of thresholded hybrid systems. Moreover, we show that, under some mild conditions regarding causality, this continuity also continues to hold for arbitrary compositions, 
which in turn guarantees that the composition faithfully captures the analog reality. By applying our generic results to some recently developed thresholded hybrid gate models, both for single-input single-output gates like inverters and for a two-input CMOS \NOR\ gate, we show that they are continuous. Moreover, we 
provide a novel thresholded hybrid model for the two-input \NOR\ gate, which is not only continuous but also, unlike the existing one, faithfully models all multi-input switching effects.
\end{abstract}


\end{frontmatter}



\section{Introduction}
\label{sec:intro}

The behavior of \emph{thresholded hybrid systems} is governed by the dynamics of 
a continuous process, described by some system of \emph{ordinary differential equations} 
(ODEs), which is selected according to externally supplied digital
mode switch signals from a set of candidates, and controls some digital outputs
based on whether some internal signals are above or below a threshold, 
see \cref{fig:switched} for an illustration. Thresholded hybrid systems
can be found in various application areas, including digitally controlled 
thermodynamic processes, hydrodynamic systems, and, in particular, 
digital integrated circuits. Consider a simple digitally controlled heating system, 
for example: The continuous dynamics of the room temperature would be governed by some ODE
for the case when the heating is switched on, and another ODE for the case where
the heating is switched off. A binary mode switch input signal tells whether the
heating is switched on or off. Two binary output signal, low resp.\ high, report on 
whether the current room temperature is below $20$ degrees resp.\ above $23$ degrees. 
A simple digital bang-bang controller could be used to switch the heating on when low makes a transition from $0$ to $1$, and to switch it off when high makes such a transition.

\begin{figure}[h!]
  \centerline{
    \includegraphics[width=0.63\linewidth]{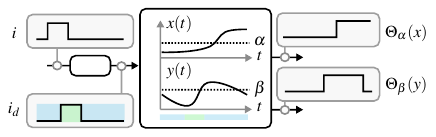}
  }
  \caption{Thresholded mode-switched ODE with a single mode input $i$, the delayed
    input $i_d$, two continuous
    states $x,y$, and two thresholded outputs $\Theta_\alpha(x)$ and $\Theta_\beta(y)$.}
  \label{fig:switched}
\end{figure}

In this paper, we will study properties of such thresholded hybrid systems, 
and systems built from those via arbitrary composition, i.e., where the 
digital output signal of one component drives a mode switch signal of 
one or more other components, possibly forming
some feedback loops. First and foremost, we will give conditions that ensure the continuity of the 
outputs of such systems with respect to their external inputs. This continuity 
property guarantees that small timing variations of the digital input signals 
lead to small variations of the digital output signals only. Moreover, we will
show that, under some mild constraints regarding causality, any finite 
composition of continuous thresholded hybrid systems is continuous, and 
faithfully models the analog reality.

Whereas our continuity results are independent of the particular application area, 
we will tailor our presentation primariy to digital integrated circuits.
Indeed, digital circuits are a particularly important class of systems composed
of thresholded hybrid systems, modeling elementary logic gates like inverters,
\GOR, and \Gand, which will be called \emph{digitized hybrid gates} in the sequel.
The application of our generic results will reveal that thresholded hybrid systems are indeed ideally suited for making the transition from the analog implementation to the digital abstraction in modern digital circuit design rigorous.

\paragraph*{Digital circuit modeling basics}
Modern digital integrated circuits consist of Millions and sometimes Billions 
of transistors, which are analog electronic devices and thus process and generate
analog signals. Modern digital circuit design, on the other hand, considers 
a circuit as a composition of elementary 
digital logic gates, and leaves it to (quite complex) tools to compile a 
design down to its analog implementation.

In view of the very short design cycles nowadays, developers cannot afford
to repeatedly downcompile a design to verify its correctness and 
performance. Fast \emph{digital} functional verification and 
timing analysis techniques and tools are hence key elements of modern circuit 
design. In particular, thanks to the elaborate static timing analysis techniques 
available today, worst-case critical path delays can be determined very accurately 
and very quickly, even for very large circuits. Whereas such corner-case delay 
estimates are sufficient for synchronous circuit designs, which are still 
the vast majority nowadays, analyzing the behavior of specific asynchronous circuits, 
like the one described in \cite{CharlieEffect}, or inter-neuron links 
using time-based encoding in hardware-implemented spiking neural 
networks~\cite{BVMRRVB19}, require dynamic timing analysis techniques.

Since analog simulations of downcompiled designs are prohibitively time-consuming,
\emph{digital} dynamic timing analysis techniques have been invented as a less accurate 
but much faster alternative. They rest on 
fast and efficiently computable \emph{gate delay models} like pure or
inertial delays \cite{Ung71}, which provide
input-to-output delay estimations for every gate. 
Since the gate delay for a given signal transition 
is also dependent on the previous transition(s), however, 
in particular, when they are close, \emph{single-history delay models} 
like \cite{FNS16:ToC,BJV06,FNNS19:TCAD} have been proposed, where 
the input-to-output delay $\delta(T)$ of a gate depends on the 
previous-output-to-input delay~$T$.

It has been proved by F\"ugger et al.~\cite{FNS16:ToC,FNNS19:TCAD} that 
\emph{continuity} is mandatory for any single-history model of a gate to 
faithfully represent the analog reality.
Continuity ensures, for example, that a constant-low input signal and an arbitrarily 
short low-high-low pulse lead to very similar gate output signals.
Note that this continuity property also implies continuity
of the output signal power w.r.t.\ the input signal power,
since the square of a signal is proportional to its power.
Consequently, continuous delay models are the most promising 
candidates for the timing and power-accurate simulation of 
digital circuits \cite{najm1994survey}.

So far, the only delay model that is known to ensure continuity
is the \emph{involution delay model} (IDM) \cite{FNNS19:TCAD}, 
which consists of zero-time Boolean gates interconnected by single-input
  single-output involution delay channels.
An IDM channel is characterized by a delay function~$\delta$
  that is a negative involution, i.e., $-\delta(-\delta(T))=T$.
In its generalized version, different delay functions $\dup$ resp.\ $\ddo$
  are assumed for rising resp.\ falling transitions, requiring 
  $-\dup(-\ddo(T))=T$.
The involution property happens to guarantee continuity, which in turn is 
the key to proving that the IDM allows to solve the canonical 
\emph{short-pulse filtration problem} (see \cref{def:SPF}) 
exactly as it is possible with real circuits.

It has been shown already in~\cite{FNNS19:TCAD} that involution delay functions 
arise naturally in the 2-state thresholded hybrid model illustrated in \cref{fig:analogy}, which
consists of a pure delay component, a slew-rate limiter with a rising and falling switching
waveform, and an ideal comparator:
The binary-valued input~$i_a$ is delayed by some $\dmin>0$, which assures 
causality, i.e., $\delta_{\uparrow/\downarrow}(0)>0$.
At every transition of~$i_d$, the slew-rate limiter switches to the corresponding waveform ($\fdo/\fup$ for a falling/rising transition), thereby ensuring
that the resulting analog output voltage~$o_a$ is 
a \emph{continuous} (but not necessarily smooth) function of time.
Finally, the comparator generates the output~$o_d$ by
  digitizing~$o_a$ w.r.t.\ the discretization threshold voltage~$\vth$.

\begin{figure}[t!]
  \centering
  \subfloat[Thresholded hybrid model]{
    \includegraphics[width=0.48\linewidth]{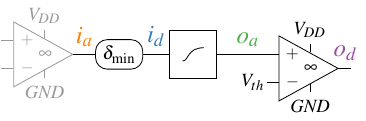}%
    \label{corFig3}}
  \hfil
  \subfloat[Sample execution]{
    \includegraphics[width=0.48\linewidth]{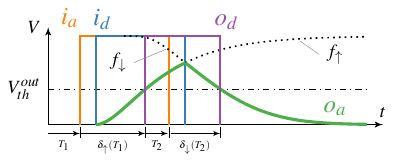}%
    \label{corFig5}}
  \caption{A digitized hybrid gate model (for a non-inverting buffer) satisfying the involution property and 
a sample execution. Adapted from~\cite{FNNS19:TCAD}.}
  \label{fig:analogy}
\end{figure}

\medskip

Whereas the accuracy of IDM predictions for single-input, single-output circuits 
like inverter chains or clock trees turned out to be very good~\cite{OMFS20:INTEGRATION}, this is less so for circuits involving multi-input gates.
It has been revealed by Ferdowsi et al.~\cite{FMOS22:DATE} that this is primarily due to the IDM's inherent inability to properly cover
  output delay variations caused by \emph{multiple input switching} (MIS) effects, also known as \emph{Charlie effects}, where different inputs switch in close temporal proximity~\cite{CGB01:DAC}: compared to the \emph{single input switching} (SIS) case, output transitions may be sped up/slowed down with decreasing
  transition separation time on different inputs.
Since circuit models based on single-input, single-output delay channels 
like IDM inherently cannot model MIS effects, generalized delay models like
the ones presented in \cref{sec:MIG} are needed for the accurate digital modeling of multi-input gates. 

\paragraph*{Detailed contributions}
\begin{enumerate}
\item[(1)] We show that any thresholded hybrid model, where mode $m$ 
  is governed by a system of first-order ODEs $\frac{dx}{dt} = F_m(t,x)$,
  leads to a continuous digital delay model, provided all the $F_m$ are 
  continuous in $t$ and Lipschitz continuous in $x$, with a common 
  Lipschitz constant for every $t>0$ and $m$. 

\item[(2)] We carry over our general continuity property to
  digitized hybrid gates.

\item[(3)] We prove that the parallel composition of finitely many digitized hybrid 
gates in a circuit result in a unique and Zeno-free execution, under some mild conditions
regarding causality. Moreover, we prove that the resulting model is faithful 
w.r.t.\ solving the canonical short-pulse filtration problem, provided
all involved digitized hybrid gates are continuous.

\item[(4)] We introduce the intricacies caused by MIS effects in multi-input gates,
and show that the digitized hybrid model for CMOS \NOR\ gates proposed in 
\cite{FMOS22:DATE} is continuous. 
  
\item[(5)] We revisit the advanced digitized hybrid mode for CMOS \NOR\ gates
presented in \cite{ferdowsi2023accurate}, which covers all MIS effects. We prove that it is 
continuous, and derive an accurate approximation of its 
delay function based on explicit solutions of the underlying ODEs.\footnote{A note 
to the reviewers: The present paper combines the HSCC'23 
paper \cite{ferdowsi2023continuity} (where we presented our continuity proof) and the 
ICCAD'23 paper \cite{ferdowsi2023accurate} (where we presented our advanced 
model for the \NOR\ gate), with the important difference that we replace the complicated
approximation of the ODE solutions used in \cite{ferdowsi2023accurate} by the 
recently found explicit solutions, which results in much simpler and more accurate 
delay formulas and an explicit model parametrization procedure that avoids any fitting.}
\end{enumerate}

\paragraph*{Paper organization}
In \cref{sec:continuity}, we instantiate our general continuity result (\cref{thm:thr:cont}).
\cref{sec:gatemodels} presents our main continuity results for digitized
hybrid gates (\cref{thm:channels:are:cont} and \cref{thm:digitizedmodels:are:cont}), and
\cref{sec:circuits} deals with circuit composition and faithfulness of composed models.
In \cref{sec:MIG}, we introduce MIS effects in multi-input gates and apply our 
continuity and faithfulness results to existing digitized hybrid models \cite{FNNS19:TCAD,FMOS22:DATE}.
In the comprehensive \cref{sec:AdvancedModel}, we provide our novel analysis
of the advanced model introduced in \cite{ferdowsi2023accurate}. 
Some conclusions are provided in \cref{sec:conclusions}.

\section{Thresholded Mode-Switched ODEs}
\label{sec:continuity}
In this section, we provide a generic proof that every hybrid model that adheres to some mild
  conditions on its ODEs leads to a continuous digital delay model.
We start with proving continuity in the analog domain and then establish continuity
  of the digitized signal obtained by feeding a continuous real-valued
  signal into a threshold voltage comparator.
Combining those results will allow us to assert the continuity of digital delay channels
  like the one shown in \cref{fig:analogy}.

\subsection{Continuity of ODE mode switching}
\label{sec:contODE}
For a vector $x\in\IR^n$, denote by~$\lVert x\rVert$ its Euclidean norm.
For a piecewise continuous function $f:[a,b]\to\IR^n$, we write $\lVert f\rVert_1 = \int_a^b \lVert f(t) \rVert \, dt$ for its $1$-norm and
$\lVert f\rVert_\infty = \sup_{t\in[a,b]} \lVert f(t)\rVert$ for its supremum norm.
The projection function of a vector in $\IR^n$ onto its $k^\text{th}$ component, for $1 \leq k \leq n$,
is denoted by $\pi_k : \IR^n \to \IR$. 

In this section, we will consider non-autonomous first-order ODEs of the form 
$\frac{d}{dt}\,x(t) = f(t,x(t))$, where the non-negative $t \in \IR_+$ represents
the time parameter, $x(t) \in U $ for some arbitrary open set $U\subseteq \IR^n$,
$x_0 \in U$ is some initial value, and $f:\IR_+ \times U\to\IR^n$ is chosen
from a set $F$ of bounded functions that are continuous for $(t,x) \in [0,T] \times U$,
where $0 < T < \infty$, and Lipschitz continuous in $U$ with 
a common Lipschitz constant for all $t\in[0,T]$ and all choices 
of $f \in F$. It is well-known that every such ODE has a unique solution $x(t)$ 
with $x(0)=x_0$ that satisfies $x(t)\in U$ for $t\in[0,T]$, is continuous in $[0,T]$, 
and differentiable in $(0,T)$.

The following lemma shows the continuous dependence of the solutions of such ODEs 
on their initial values. To be more explicit, the exponential dependence of the Lipschitz constant on the time parameter allows temporal composition of the bound.
The proof can be found in standard textbooks on ODEs~\cite[Theorem~2.8]{teschl2012ordinary}.

\begin{lem}\label{lem:cont:wrt:initial:value}
Let $U\subseteq \IR^n$ be an open set and let $f:\IR \times U\to\IR^n$ be Lipschitz continuous with Lipschitz constant~$K$ for $t\in[0,T]$ with $T > 0$, and let $x,y:[0,T]\to U$ be continuous functions that are differentiable on $(0,T)$ such that $\frac{d}{dt}\,x(t) = f(t,x(t))$
and $\frac{d}{dt}\,y(t) = f(t,y(t))$
for all $t\in (0,T)$. Then, $\lVert x(t) - y(t) \rVert
\leq
e^{t K} \lVert x(0) - y(0) \rVert$ for all $t\in[0,T]$.
\end{lem}

A \emph{step function} $s:\IR_+ \to \{0,1\}$ is a right-continuous function 
with left limits, i.e., $\lim_{t\to t_0^+} s(t)=s(t_0)$ and $\lim_{t\to t_0^-} s(t)$
exists for all $t_0 \in \IR_+$.
A \emph{binary signal} $s$ is a step function $s:[0,T]\to\{0,1\}$,
a \emph{mode-switch signal} $a$ is a step function $a:[0,T]\to F$, $t\mapsto a_t$.

Given a mode-switch signal~$a$, a \emph{matching output signal} for~$a$ is 
a function $x_a:[0,T]\to U$ that satisfies
\begin{enumerate}
\item[(i)] $x_a(0) = x_0$,
\item[(ii)] the function~$x_a$ is continuous,
\item[(iii)] for all $t\in(0,T)$, if~$a$ is continuous at~$t$, then~$x_a$ is differentiable at~$t$ and $\frac{d}{dt}\,x_a(t) = a_t(t,x_a(t))$.
\end{enumerate}
For (iii), recall that the domain of $a$ is $F$.

\begin{lem}[Existence and uniqueness of matching output signal] \label{lem:uniqueness}
Given a mode-switch signal~$a$, the matching output signal $x_a$ for~$a$
exists and is unique.
\end{lem}

\begin{proof}
$x_a$ can be constructed inductively, by pasting together the solutions 
$x_{t_j}$ of $\frac{d}{dt}\,x_{t_j}(t) = a_{t_j}(t,x_{t_j}(t))$, where $t_0=0$
and $t_1 < t_2 < \dots$ are $a$'s switching times in $S_a$: 
For the induction basis $j=0$, we define $x_a(t):=x_{t_0}(t)$
with initial value $x_{t_0}=x_{t_0}(t_0):=x_0$ for $t\in [0,t_1]$. Obviously, (i) holds
by construction, and the continuity and differentiability of $x_{t_0}(t)$ at other times
ensures (ii) and (iii).

For the induction step $j \to j+1$, we assume 
that we have constructed $x_a(t)$ already for $0 \leq t \leq t_j$. For $t\in [t_j,t_{j+1}]$, we 
define $x_a(t):=x_{t_{j+1}}(t)$ with initial value $x_{t_{j+1}}=x_{t_{j+1}}:=x_a(t_j)=x_{t_j}(t_j)$.
Continuity of $x_a(t)$ at $t=t_j$ follows by construction, and the continuity and differentiability 
of $x_{t_{j+1}}(t)$ again ensures~(ii) and~(iii).
\end{proof}

Given two mode-switch signals~$a$, $b$, we define their distance as
\begin{equation}
d_T(a,b)
=
\lambda\big(\{ t\in[0,T] \mid a_t \neq b_t \}\big)
\end{equation}
where~$\lambda$ is the Lebesgue measure on~$\IR$.
The distance function~$d_T$ is a metric on the set of mode-switch signals.

The following \cref{thm:real-valued:is:cont} shows that the mapping 
$a\mapsto x_a$ is continuous.

\begin{thm}\label{thm:real-valued:is:cont}
Let~$K\geq 1$ be a common Lipschitz constant for all functions in~$F$ and let~$M$ be a real number such that $\lVert f(t,x(t))\rVert \leq M$ for all $f\in F$, all $x\in U$, and all $t \in [0,T]$.
Then, for all mode-switch signals~$a$ and~$b$,
if~$x_a$ is the output signal for~$a$
and~$x_b$ is the output signal for~$b$,
then
$\lVert x_a - x_b\rVert_{\infty} \leq 2Me^{TK} d_T(a,b)$. Consequently, the mapping $a\mapsto x_a$ is continuous.
\end{thm}
\begin{proof}
Let $S = \{ t\in (0,T) \mid \text{$a$ or $b$ is discontinuous at $t$} \} \cup \{0,T\}$ be the set of switching times of~$a$ and~$b$.
The set~$S$ must be finite, since both~$a$ and~$b$ are right-continuous on a compact interval.
Let $0=s_0 < s_1 < s_2 < \cdots < s_m=T$ be the increasing enumeration of~$S$.

We show by induction on~$k$ that
\begin{equation}\label{eq:thm:real-valued:is:cont:induction}
\forall t\in [0, s_k] \colon\quad
\lVert x_a(t) - x_b(t) \rVert \leq 2Me^{t K} d_t(a,b)
\end{equation}
for all $k\in\{0,1,2,\dots,m\}$.
The base case $k=0$ is trivial.
For the induction step $k \mapsto k+1$, we distinguish the two cases
$a_{s_k} = b_{s_k}$ and
$a_{s_k} \neq b_{s_k}$.

If $a_{s_k} = b_{s_k}$, then we have $a_t = b_t$ for all $t\in [s_k,s_{k+1})$ and
hence $d_t(a,b) = d_{s_k}(a,b)$ for all $t\in[s_k,s_{k+1}]$.
Moreover, we can apply Lemma~\ref{lem:cont:wrt:initial:value} and obtain
\begin{equation}
\forall t\in [s_k, s_{k+1}]\colon\quad
\lVert x_a(t) - x_b(t) \rVert
\leq
e^{(t - s_k)K} \lVert x_a(s_k) - x_b(s_k) \rVert
\enspace.
\end{equation}
Plugging in~\eqref{eq:thm:real-valued:is:cont:induction} for $t = s_k$ reveals that~\eqref{eq:thm:real-valued:is:cont:induction} holds for all $t\in[s_k, s_{k+1}]$ as well.

If $a_{s_k} \neq b_{s_k}$, then~$x_a$ and~$x_b$ follow different differential equations in the interval $t \in [s_k,s_{k+1}]$.
We can, however, use the mean-value theorem for vector-valued functions~\cite[Theorem~5.19]{rudin1976principles} to obtain
\begin{equation}
\forall t\in [s_k, s_{k+1}]\colon\quad
\lVert x_a(t) - x_a(s_k) \rVert
\leq
M(t - s_k)\enspace\text{and}
\end{equation}
\begin{equation}
\forall t\in [s_k, s_{k+1}]\colon\quad
\lVert x_b(t) - x_b(s_k) \rVert
\leq
M(t - s_k).
\end{equation}
This, combined with the induction hypothesis, the equality $d_t(a,b) = d_{s_k}(a,b) + (t-s_k)$, and the inequalities $1\leq e^{t K}$ and $e^{s_k K} \leq e^{t K}$, implies
\begin{equation*}
\begin{split}
\lVert x_a(t) - x_b(t) \rVert
& \leq
\lVert x_a(t) - x_a(s_k) + \lVert x_a(s_k) - x_b(s_k) \rVert
+
\lVert x_b(s_k) - x_b(t) \rVert
\\ & \leq
2M(t-s_k)
+
2Me^{s_k K} d_{s_k}(a,b)
\\ & \leq
2M e^{t K} (t-s_k)
+
2Me^{t K} d_{s_k}(a,b)
\\ & =
2M e^{t K} \big( d_t(a,b) - d_{s_k}(a,b) \big)
+
2Me^{t K} d_{s_k}(a,b)
\\ & =
2M e^{t K} d_t(a,b)
\end{split}
\end{equation*}
for all $t\in [s_k, s_{k+1}]$.
This concludes the proof.
\end{proof}

We remark that the (proof of the)
continuity property
of \cref{thm:real-valued:is:cont} is very different from the standard
(proof of the) continuity property of controlled variables in
closed thresholded hybrid systems. Mode switches in such systems
are caused by the time evolution of the system
itself, e.g., when some controlled variable exceeds some value. Consequently,
such systems can be described by means of a \emph{single} ODE system with
discontinuous righthand side \cite{Fil88}.

By contrast, in our hybrid systems, the mode switches are solely caused by changes
of digital inputs that are \emph{externally} controlled: For every possible pattern of
the digital inputs, there is a dedicated ODE system that controls the 
analog output. Consequently, the time evolution of the output now also depends 
on the time evolution of the inputs. Proving the
continuity of the (discretized) output w.r.t.\ different 
(but close, w.r.t. some metric) digital input signals require
relating the output of \emph{different} ODE systems.

\subsection{Continuity of thresholding}
\label{sec:thresholding}
For a real number $\xi\in\IR$ and a function $x:[a,b]\to\IR$, denote by $\Theta_\xi(x)$ the thresholded version of $x$ defined by
\begin{equation}
\Theta_\xi(x) : [a,b] \to \{0,1\},
\quad
\Theta_\xi(x)(t) =
\begin{cases}
0 & \text{if } x(t) \leq \xi,\\
1 & \text{if } x(t) > \xi.
\end{cases}
\end{equation}

\begin{lem}\label{lem:thr:cont:monotonic:end}
Let $\xi\in\IR$ and let $x:[a,b]\to\IR$ be a continuous strictly monotonic function with $x(b)=\xi$.
Then, for every $\varepsilon>0$, there exists a $\delta > 0$ such that,
for every continuous function $y:[a,b]\to\IR$, the condition
$\lVert x - y\rVert_\infty < \delta$ implies $\lVert \Theta_\xi(x) - \Theta_\xi(y)\rVert_1 < \varepsilon$.
\end{lem}
\begin{proof}
We show the lemma for the case that~$x$ is strictly increasing.
The proof for strictly decreasing~$x$ is analogous.

Set $\chi = x(a)$.
Since~$x$ is bijective onto the interval $[\chi,\xi]$, it has an inverse function $x^{-1}:[\chi,\xi]\to [a,b]$.
The inverse function~$x^{-1}$ is continuous because the domain $[a,b]$ is compact \cite[Theorem~4.17]{rudin1976principles}.

The relation $t \leq x^{-1}(\xi - \delta)$ implies $x(t)+\delta \leq \xi$.
Hence, if $\lVert x - y\rVert_\infty < \delta$, then $y(t) \leq x(t) + \delta \leq \xi$ for all $t\leq x^{-1}(\xi-\delta)$.
This means that $\Theta_\xi(y)(t) = 0$ for all $t\leq x^{-1}(\xi-\delta)$,
so $t>x^{-1}(\xi-\delta)$ for every $t\in[a,b]$ where $\Theta_\xi(y)(t) = 1$.

By assumption, we have $\Theta_\xi(x)(t) = 0$ for all $t\in [a,b]$.
Thus,

\begin{flalign}
\lVert \Theta_\xi(x) - \Theta_\xi(y)\rVert_1 & =\lambda\big( \{ t\in[0,T] \mid \Theta_\xi(y) = 1 \} \big)=\lambda\big( \{ t\in[0,T] \mid y(t) > \xi \} \big)  \nonumber &\\
& \leq b - x^{-1}(\xi-\delta). &
\label{eq:measure1}
\end{flalign}

Note that continuity of $y$ is sufficient to ensure that the set in \cref{eq:measure1} is measurable.
Since~$x^{-1}$ is continuous, we have $x^{-1}(\xi-\delta) \to x^{-1}(\xi) = b$ as $\delta \to 0$.
In particular, for every $\varepsilon > 0$, there exists a $\delta > 0$ such that $b - x^{-1}(\xi-\delta) < \varepsilon$.
This concludes the proof.
\end{proof}

The following \cref{lem:thr:cont:monotonic} shows that we can drop the assumption $x(b)=\xi$
in \cref{lem:thr:cont:monotonic:end}:

\begin{lem}\label{lem:thr:cont:monotonic}
Let $\xi\in\IR$ and let $x,y:[a,b]\to\IR$ be two continuous functions where $x$ is strictly monotonic. Then, for every $\varepsilon>0$, there exists a $\delta > 0$ such that $\lVert x - y\rVert_\infty < \delta$ implies $\lVert \Theta_\xi(x) - \Theta_\xi(y)\rVert_1 < \varepsilon$. Moreover, if $\max\{x(a),x(b)\}\leq \xi$ or
$\min\{x(a),x(b)\}> \xi$, then $\lVert \Theta_\xi(x) - \Theta_\xi(y)\rVert_1 =0$.
\end{lem}
\begin{proof}
We again show the lemma for the case that~$x$ is strictly increasing.
The proof for strictly decreasing~$x$ is analogous.

Let $\varepsilon > 0$. We distinguish three cases:

(i) If $x(b) < \xi$, then we have $\Theta_\xi(x)(t)=0$ for all $t\in[a,b]$.
Choosing $\delta = \xi - x(b)$, we deduce $y(t) < x(t) + \delta \leq x(b) + \xi - x(b) = \xi$ for all $t\in[a,b]$ whenever $\lVert x - y\rVert_\infty < \delta$.
Hence, we get $\Theta_\xi(y)(t)=0$ for all $t\in[a,b]$ and thus
$\lVert \Theta_\xi(x) - \Theta_\xi(y)\rVert_1 = 0 < \varepsilon$.

(ii) If $x(a) > \xi$, then we can choose $\delta = x(a) - \xi$ and get $\Theta_\xi(y)(t) = \Theta_\xi(x)(t) = 1$ for all $t\in[a,b]$ whenever $\lVert x-y \rVert_\infty < \delta$.
In particular, $\lVert \Theta_\xi(x) - \Theta_\xi(y)\rVert_1 = 0 < \varepsilon$.

(iii) If $x(a) \leq \xi \leq x(b)$, then there exists a unique $c\in[a,b]$ with $x(c) = \xi$.
Applying Lemma~\ref{lem:thr:cont:monotonic:end} on the restriction of~$x$ on the interval $[a,c]$, we get the existence of a $\delta_1>0$ such that $\lVert x- y\rVert_{[a,c],\infty} < \delta_1$ implies $\lVert \Theta_\xi(x) - \Theta_\xi(y)\rVert_{[a,c],1} < \varepsilon/2$; herein, $\lVert \cdot \rVert_{[a,c],\infty}$ and $\lVert \cdot \rVert_{[a,c],1}$ denote the supremum-norm and the $1$-norm on the interval~$[a,c]$, respectively.
Applying Lemma~\ref{lem:thr:cont:monotonic:end} on the restriction of~$x$ on the interval~$[c,b]$ after the coordinate transformation $t\mapsto -t$ yields the existence of a $\delta_2>0$ such that $\lVert x- y\rVert_{[c,b],\infty} < \delta_2$ implies $\lVert \Theta_\xi(x) - \Theta_\xi(y)\rVert_{[c,b],1} < \varepsilon/2$.
Setting $\delta = \min\{\delta_1,\delta_2\}$, we thus get
$\lVert \Theta_\xi(x) - \Theta_\xi(y)\rVert_{[a,b],1}
 =
\lVert \Theta_\xi(x) - \Theta_\xi(y)\rVert_{[a,c],1}
+
\lVert \Theta_\xi(x) - \Theta_\xi(y)\rVert_{[c,b],1}
 <
\varepsilon/2 + \varepsilon/2
=
\varepsilon
$
whenever $\lVert x - y\rVert_{[a,b],\infty} < \delta$.
\end{proof}

The following \cref{thm:thr:cont} shows that the mapping $x\mapsto \Theta_\xi(x)$ 
is continuous for a given function $x$, provided that $x$ has only finitely many
alternating critical points, i.e., local optima that alternate between lying above and below $\xi$. Formally, these are times $t_0, t_1, \dots, t_M$ where $x'(t_j)=0$ 
for all $0\leq j \leq M$ and $\sgn\bigl(x(t_i)-\xi\bigr)=-\sgn\bigl(x(t_{i+1})-\xi\bigr)$, for all $0\leq i \leq M-1$. Note carefully that we require $M$
to be fixed and hence, in particular, independent of the choice of $T$ here.

\begin{thm}\label{thm:thr:cont}
Let $\xi\in\IR$ and let $x,y:[0,T]\to\IR$ be two differentiable functions.
Assume that~$x$ has at most $M<\infty$ alternating critical points, where
$M$ is independent of $T$.
Then, for every $\varepsilon>0$, there exists a $\delta > 0$ such that $\lVert x - y\rVert_\infty < \delta$ implies $\lVert \Theta_\xi(x) - \Theta_\xi(y)\rVert_1 < \varepsilon$. Consequently, the mapping $x\mapsto \Theta_\xi(x)$ is continuous.
\end{thm}
\begin{proof}
Let $\mathcal{N} = \left\{t\in[0,T] \mid \text{$x$ has a critical point at $t$}\right\}\cup\{0,T\}$, which contains only $m\leq M$ alternating critical points by assumption, and 
$t_0 < t_1 < t_2 < \cdots < t_m$ be the increasing enumeration of~$\mathcal{N}$.
By the mean-value theorem, the function~$x$ is strictly monotonic in every interval $[t_k,t_{k+1}]$ for every $k\in\{0,1,2,\dots,m\}$.

Let $\varepsilon>0$.
Applying Lemma~\ref{lem:thr:cont:monotonic} to the restriction of~$x$ on each of the intervals $[t_k,t_{k+1}]$, we distinguish two cases: (i) if $t_k, t_{k+1}$ are
non-alternating critical points, then $\lVert \Theta_\xi(x) - \Theta_\xi(y)\rVert_{[t_k,t_{k+1}],1}=0$. Otherwise, we are assured of the existence of some $\delta_k>0$ such that $\lVert x - y\rVert_{[t_k,t_{k+1}],\infty} < \delta_k$ implies
$\lVert \Theta_\xi(x) - \Theta_\xi(y)\rVert_{[t_k,t_{k+1}],1} < \varepsilon/m$.
Setting $\delta = \min\{\delta_{k_0}, \delta_{k_1}, \delta_{k_2}, \dots, \delta_{k_{m-1}}\}$, we thus obtain
\begin{equation}
\begin{split}
\lVert \Theta_\xi(x) - \Theta_\xi(y)\rVert_{[0,T],1}
& =
\sum_{k=0}^{m-1}
\lVert \Theta_\xi(x) - \Theta_\xi(y)\rVert_{[t_k,t_{k+1}],1} <
\sum_{k=0}^{m-1}
\varepsilon/m
=
\varepsilon,
\end{split}
\end{equation}
whenever $\lVert x - y\rVert_{[0,T],\infty} < \delta$.
\end{proof}

\section{Continuity of Digitized Hybrid Gates}
\label{sec:gatemodels}
To prepare for our general result about the continuity of hybrid gate models (\cref{thm:digitizedmodels:are:cont}), we will first
(re)prove the continuity of IDM channels as shown in \cref{fig:analogy}, which
has been established by a (somewhat tedious) direct proof in \cite{FNNS19:TCAD}. 
In our notation, an IDM channel consists of:
\begin{itemize}
\item A nonnegative minimum delay $\dmin\geq0$ and a delay function $\Delta_{\dmin}(s)$ that maps the binary input signal 
$i_a$, augmented with the left-sided limit $i_a(0-)$ as the \emph{initial value}\footnote{In \cite{FNNS19:TCAD}, this initial
value of a signal was encoded by extending the time domain to the whole $\IR$ and using $i_a(-\infty)$.} that can be different from $i_a(0)$, to the binary signal $i_d=\Delta_{\dmin}(i_a)$, defined by
\begin{equation}
\Delta_{\dmin}(i_a)(t) =
\begin{cases}
i_a(0-) & \text{if } t< \dmin\\
i_a(t-\dmin) & \text{if } t\geq \dmin
\enspace.
\end{cases}
\end{equation}
\item An open set $U\subseteq \IR^n$, where $\pi_1[U]$ represents the analog output signal and $\pi_k[U]$, $k=\{2,3, \ldots, n \}$, specifies the internal state variables of the model. In this fashion,\footnote{In real circuits, the
interval $(0,1)$ typically needs to be replaced by $(0,\vdd)$.} we presume that $\pi_1[U] = (0,1)$, i.e., the range of output signals is contained in the interval $(0,1)$.
\item Two bounded functions $F_\uparrow, F_\downarrow :\IR \times U \to \IR^n$ 
 with the following properties:
    \begin{itemize}
        \item $F_\uparrow, F_\downarrow$ are continuous for $(t,x) \in [0,T] \times U$, 
         for any $0 < T < \infty$, and Lipschitz continuous in $U$, which entails
that every trajectory~$x$ of the ODEs $\frac{d}{dt}\,x(t) = F_\uparrow(t,x(t))$ and $\frac{d}{dt}\,x(t) = F_\downarrow(t,x(t))$ with any initial value $x(0)\in U$ satisfies $x(t)\in U$ for all~$t\in [0,T]$, recall \cref{sec:contODE}.
        \item for no trajectory~$x$ of the ODEs $\frac{d}{dt}\,x(t) = F_\uparrow(t,x(t))$ and $\frac{d}{dt}\,x(t) = F_\downarrow(t,x(t))$ with initial value $x(0)\in U$ does $\pi_1[x]$ have infinitely many alternating critical points $t_0, t_1, \dots$ with $\pi_1[x]'(t)=0$ and $\sgn\bigl(\pi_1[x](t_i)-\xi\bigr)=-\sgn\bigl(\pi_1[x](t_{i+1})-\xi\bigr)$, for all $i\geq 0$.
    \end{itemize}
\item An initial value $x_0\in U$, with $x_0=F_\uparrow$ if $i_a(0-) = 1$ and $x_0 = F_\downarrow$ if $i_a(0-)=0$.
\item A mode-switch signal $b:[0,T] \to\{F_\uparrow,F_\downarrow\}$ defined by setting $b(t) = F_\uparrow$ if $i_d(t) = 1$ and $b(t) = F_\downarrow$ if $i_d(t)=0$.
\item The analog output signal $o_a=\pi_1[x_{b}]$, i.e., the output signal for~$b$ and initial value~$x_0$.
\item A threshold voltage $\xi=\vth \in(0,1)$ for the comparator that finally produces the binary output
signal $o_d=\Theta_\xi(o_a)$.
\end{itemize}

By combining the results from \cref{sec:contODE} and \ref{sec:thresholding}, we obtain:

\begin{thm}\label{thm:channels:are:cont}
The channel function of an IDM channel, which maps from the input signal $i_a$ to the 
output signal $o_d$, is continuous with respect to the $1$-norm on the interval~$[0,T]$.
\end{thm}
\begin{proof}
The mapping from $i_a$ to $o_d$ is continuous as the concatenation of continuous mappings:
\begin{itemize}
\item The mapping from $i_a \mapsto i_d$ is continuous since $\Delta_{\dmin}$ is trivially
continuous for input and output binary signals with the $1$-norm.
\item The mapping $i_d\mapsto b$ is a continuous mapping from the set of signals equipped with the $1$-norm to the set of mode-switch signals equipped with the metric~$d_T$, since the points of discontinuity of $b$ are the points where $i_d$ is discontinuous.
\item By Theorem~\ref{thm:real-valued:is:cont}, the mapping $b\mapsto x_{b}$ is a continuous mapping from the set of mode-switch signals equipped with the metric~$d_T$ to the set of piecewise differentiable functions $[0,T]\to U$ equipped with the supremum-norm.
\item The mapping $x_b\mapsto \pi_1\circ x_{b}$ is a continuous mapping from the set of piecewise differentiable functions $[0,T]\to U$ equipped with the supremum-norm to the set of piecewise differentiable functions $[0,T]\to (0,1)$ equipped with the supremum-norm.  Since $\lVert (x_1,\dots,x_n) \rVert_1 = 
  \lVert x_1 \rVert_1 + \dots +  \lVert x_n \rVert_1$ for every $x \in U$,
  this follows from $\lVert  \pi_1[x]\rVert_1 \leq \lVert x\rVert_1$.
\item By Theorem~\ref{thm:thr:cont}, the mapping
$\pi_1\circ x_{b} \mapsto \Theta_\xi(\pi_1\circ x_{b})$ is a continuous mapping from the set of piecewise differentiable functions $[0,T]\to (0,1)$ equipped with the supremum-norm to the set of binary signals equipped with the $1$-norm.\qedhere
\end{itemize}
\end{proof}

Whereas the condition that no trajectory of any of the ODEs may have infinitely 
many alternating critical points is difficult to check in general, it is always guaranteed
for every switching waveform $f(t)$ typically found in elementary\footnote{This is true for all combinational gates like inverters, $\NOR$, $\NAND$ etc. Excluded are
gates with an internal state, like a storage element, which may exhibit 
\emph{metastable behavior} \cite{Mar77,Mar81}.} digitized 
hybrid gates. 
More specifically, as any $f$ is meant to represent a digital signal here, it must  satisfy either $\lim_{t\to\infty} f(t)=0$ or $\lim_{t\to\infty} f(t)=1$.
Moreover, since real circuits cannot produce waveforms with arbitrary steep
slopes, $|f'(t)|$ must be bounded. From the former, it follows that,
for every $\varepsilon > 0$, there is some $t(\varepsilon)$
such that either $f(t) < \varepsilon$ or else $f(t) > 1-\varepsilon$ 
for every $t\geq t(\varepsilon)$. Consequently, choosing $\varepsilon = \min\{\vth,1-\vth\}$ reveals that no alternating critical point 
$t\geq t(\varepsilon)$ can exist. Infinitely many alternating critical points
for $t < t(\varepsilon)$ are prohibited by $|f'(t)|$ being bounded.

\medskip

With these preparations, we can now deal with the general case:
General digitized hybrid gates have $c\geq 1$ binary input signals $i_a=(i_a^1,\dots,i_a^c)$, augmented with
\emph{initial values} $(i_a^1(0-),\dots,i_a^c(0-))$, and a single binary output signal $o_d$, and are specified as follows:

\begin{dfn}[Digitized hybrid gate]\label{def:dhm}
A digitized hybrid gate with $c$ inputs consists of:
\begin{itemize}
\item $c$ delay functions $\Delta_{\delta_j}(s)$ with $\delta_j\geq 0$, $1 \leq j \leq c$, 
that map the binary input signal
$i_a^j$ with initial value $i_a^j(0-)$ to the binary signal $i_d^j=\Delta_{\delta_j}(i_a^j)$, defined by 
\begin{equation}
\Delta_{\delta_j}(i_a^j)(t)=
\begin{cases}
i_a^j(0-) & \text{if } t< \delta_j\\
i_a^j(t-\delta_j) & \text{if } t\geq \delta_j \label{eq:idj}
\enspace.
\end{cases}
\end{equation}
\item An open set $U\subseteq \IR^n$, where $\pi_1[U]$ represents the analog output signal and $\pi_k[U]$, $k=\{2,3, \ldots, n \}$, specifies the internal state variables of the model. 
\item A set $F$ of bounded functions $F^\ell :\IR \times U \to \IR^n$,
  with the following properties:
    \begin{itemize}
    \item $F^\ell$
     is continuous for $(t,x) \in [0,T] \times U$, 
          for any $0 < T < \infty$, and Lipschitz continuous in $U$, with
          a common Lipschitz constant, which entails
          that every trajectory~$x$ of the ODE $\frac{d}{dt}\,x(t) = F^\ell(t,x(t))$
          with any initial value $x(0)\in U$ satisfies $x(t)\in U$ for all~$t\in [0,T]$.
         \item for no trajectory~$x$ of the ODEs $\frac{d}{dt}\,x(t) = F^\ell(t,x(t))$ with initial value $x(0)\in U$ does $\pi_1[x]$ have infinitely many alternating critical points $t_0, t_1, \dots$ with $\pi_1[x]'(t)=0$ and $\sgn\bigl(\pi_1[x](t_i)-\xi\bigr)=-\sgn\bigl(\pi_1[x](t_{i+1})-\xi\bigr)$, for all $i\geq 0$.
     \end{itemize}
\item A mode-switch signal $b:[0,T] \to F$, which obtained by a continuous
  choice function $b_c$ acting on $i_d^1(t),\dots,i_d^c(t)$, i.e., 
  $b(t) = b_c(i_d^1(t),\dots,i_d^c(t))$. 
\item An initial value $x_0\in U$, which must correspond to the mode selected by $b_c(i_a^1(0-),\dots,i_a^c(0-))$.
\item The analog output signal $o_a=\pi_1[x_{b}]$, i.e., the output signal for~$b$ and initial value~$x_0$.
\item A threshold voltage $\xi=\vth \in(0,1)$ for the comparator that finally produces the binary output
signal $o_d=\Theta_\xi(o_a)$.
\end{itemize}
\end{dfn}

By essentially the same proof as for \cref{thm:channels:are:cont}, we obtain:

\begin{thm}\label{thm:digitizedmodels:are:cont}
The gate function of a digitized hybrid gate with $c$ inputs according to \cref{def:dhm}, which maps from the vector of input signals $i_a=(i_a^1,\dots,i_a^c)$ to the output signal $o_d$, is continuous with respect to the $1$-norm on the interval~$[0,T]$.
\end{thm}

\section{Composing Gates in Circuits}
\label{sec:circuits}
In this section, we will first compose digital circuits from digitized hybrid gates
and reason about their executions. More specifically, it will turn out that, under
certain conditions ensuring the causality of every composed gate, the resulting
circuit will exhibit a unique execution for any given execution of its
inputs. This uniqueness is mandatory for building digital dynamic
timing simulation tools.

Moreover, we will adapt the proof that no circuit with IDM
channels can solve the bounded SPF problem utilized in \cite{FNNS19:TCAD} 
to our setting: Using the continuity result of \cref{thm:digitizedmodels:are:cont}, we
prove that no circuit with digitized hybrid gates can solve bounded SPF. Since unbounded
SPF can be solved with IDM channels, which are simple instances of digitized hybrid gate
models, faithfulness w.r.t.\ solving the SPF problem follows.

\subsection{Executions of circuits}
\textbf{Circuits.}  
Circuits are obtained by interconnecting a set of input ports and a set
   of output ports, forming the external interface of a circuit, 
   and a finite set of digitized hybrid gates. 
We constrain the way components are interconnected in a natural
way, by requiring that any gate input, channel input and output
port is attached to only one input port, gate output or channel
output, respectively. 
Formally, a {\em circuit\/} is described by a directed graph where:
\begin{enumerate}
\item[C1)] A vertex $\Gamma$ can be either a {\em circuit input port}, a {\em
    circuit output port}, or a digitized hybrid {\em gate}.
\item[C2)] The \emph{edge} $(\Gamma,I,\Gamma')$ represents a $0$-delay connection from
the output of $\Gamma$ to a fixed input $I$ of $\Gamma'$. 
\item[C3)] Circuit input ports have no incoming edges.
\item[C4)] Circuit output ports have exactly one incoming edge and no outgoing one. 
\item[C5)] A $c$-ary gate $G$ has a single output and $c$ inputs $I_1,\dots, I_c$,
  in a fixed order, fed by incoming edges from exactly one gate output or input port.
\end{enumerate}

\smallskip
\noindent\textbf{Executions.}  An {\em execution\/} of a circuit~$\C$ is a collection of
binary signals~$s_\Gamma$ defined on $[0,\infty)$
for all vertices~$\Gamma$ of~$\C$ that respects all
the gate functions and input port signals.  Formally, the following
properties must hold:
\begin{enumerate}
\item[E1)] If~$i$ is a circuit input port, there are no restrictions on~$s_i$.
\item[E2)] If~$o$ is a circuit output port, then~$s_o = s_G$, where~$G$ is the
unique gate output connected to~$o$.
\item[E3)] If vertex~$G$ is a gate with~$c$ inputs $I_1,\dots,I_c$, ordered
  according to the fixed order condition C5), and gate function~$f_G$, then
  $s_G = f_G(s_{\Gamma_1},\dots, s_{\Gamma_c})$, where
  $\Gamma_1,\dots,\Gamma_c$ are the vertices the inputs $I_1,\dots,I_c$ of $C$
  are connected to via edges $(\Gamma_1,I_1,G), \dotsm, (\Gamma_d,I_c,G)$.
\end{enumerate}

The above definition of an execution of a circuit
is ``existential'', in the sense that it
  only allows checking for a given collection of signals whether 
  it is an execution or not: For every hybrid gate in the
  circuit, it specifies the gate output
  signal, given a {\em fixed\/} vector of input signals, all defined on
  the time domain $t\in[0,\infty)$.
A priori, this does not give an algorithm to construct executions of circuits,
in particular, when they contain feedback loops. Indeed, the parallel
composition of general hybrid automata may lead to non-unique executions and
bizarre timing behaviors known as \emph{Zeno}, where an infinite number of 
transitions may occur in finite time~\cite{LSV03}.

To avoid such behaviors in our setting, we require all digitized
hybrid gates in a circuit to be \emph{strictly causal}:

\begin{dfn}[Strict causality]\label{def:strictcausality}
A digitized hybrid gate $G$ with $c$ inputs is strictly causal, 
if the pure delays $\delta_j$ for every $1\leq j \leq c$ are
positive. Let $\dmin^C>0$ be the minimal pure delay of any input
of any gate in circuit $C$.
\end{dfn}

We proceed with defining input-output causality for gates,
which is based on signal transitions. Every binary signal can equivalently be described by a sequence 
of transitions: A {\em falling transition\/} at time~$t$
is the pair $(t,0)$, a {\em rising
transition\/} at time~$t$ is the pair $(t,1)$. 

\begin{dfn}[Input-output causality]\label{def:inputoutputcausality}
The output transition $(t,.)\in s_G$ of a gate G \emph{is caused}
by the transition $(t',.) \in s_G^j$ on input $I_j$ of $G$, if
$(t,.)$ occurs in the mode $a_c(i_d^1(t^+),\dots,i_d^c(t^+))$, where
$i_d^j(t^+)$ is the pure-delay shifted input signal at input $I_j$ 
at the last mode switching time $t^+\leq t$ (see \cref{eq:idj}) and
 $(t',.)$ is the last transition in $s_G^j$ before or at time
$t^+-\delta_j$, i.e., $\not\exists (t'',.) \in s_G^j$ for 
$t'<t''\leq t^+-\delta_j$. 

We also assume that the output transition $(t,.)\in s_G$
\emph{causally depends} on every transition in $s_G^j$ before 
or at time $t^+-\delta_j$.
\end{dfn}

Strictly causal gates satisfy the following obvious property:

\begin{lem}\label{lem:ioseparationtime}
If some output transition $(t,.)\in s_G$ of a strictly causal digitized hybrid gate $G$
in a circuit $C$ causally depends on its input transition $(t',.)\in s_G^j$, then $t-t' \geq \delta_j$.
\end{lem}

The following \cref{thm:execution} shows that every circuit made up of 
strictly causal gates has a unique execution, defined for $t\in[0,\infty)$.

\begin{theorem}[Unique execution]\label{thm:execution}
Every circuit $C$ made up of finitely many strictly causal digitized hybrid
gates has a unique execution, which either consists of finitely
many transitions only or else requires $[0,\infty)$ as its time
domain.
\end{theorem}
\begin{proof}
We will inductively construct this unique execution by a sequence of iterations
$\ell \geq 1$ of a simple deterministic simulation algorithm, which 
determines the prefix of the sought execution up to time $t_\ell$. 
Iteration $\ell$ deals with transitions occurring at time $t_\ell$,
starting with $t_1=0$. To every transition $e$ generated throughout
its iterations, we also assign a \emph{causal depth} $d(e)$ that 
gives the maximum causal distance to an input port: $d(e)=0$ if 
$e$ is a transition at some input port, and $d(e)$ is the maximum
of $1 + d(e^j)$, $1 \leq j \leq c$, for every transition added at the output of 
a $c$-ary gate caused by transitions $e^j$ at its inputs.

Induction basis $\ell=1$: At the beginning of iteration 1, which deals with
all transitions occurring at time $t_1=0$,
all gates are in their initial mode, which is determined by the initial
values of their inputs. They are either connected to input ports, in which
case $s_i(0-)$ is used, or to the output port of some gate $G$, in which case
$s_G(0)$ (determined by the initial mode of $G$) is used. 
Depending on whether $s_i(0-)=s_i(0)$ or not, there is also an input transition $(0,s_i(0)) 
\in s_i$ or not. All transitions in the so generated execution prefix 
$[0,t_1]=[0,0]$ have a causal depth of 0.

Still, the transitions that have happened by time $t_1$ may cause additional \emph{potential
future transitions}. They are called future transitions, because they occur only 
after $t_1$, and potential because they need not occur in the final execution. 
In particular, if there is an input transition $(0,s_i(0)) 
\in s_i$, it may cause a mode switch of every gate $G$ that 
is connected to the input port $i$. Due to \cref{lem:ioseparationtime},
however, such a mode switch, and hence each of the output 
transitions~$e$ that may occur during that new mode (which all are assigned
a causal depth $d(e)=1$), of $G$ can only 
happen at or after time $t_1+\dmin^C$.
In addition, the initial mode of any gate $G$ that is not mode switched
may also cause output transitions $e$ at arbitrary times $t > 0$, which are assigned
a causal depth $d(e)=0$. Since at most finitely many critical points may exist
for every mode's trajectory, it follows that at most \emph{finitely} many 
such future potential transitions could be generated in each of the finitely 
many gates in the circuit.
Let $t_2>t_1$ denote the time of the closest transition among all
input port transitions and all the potential future transitions just introduced.

Induction step $\ell \to \ell+1$: Assume that the execution prefix
for $[0,t_\ell]$ has already been constructed in iterations 
$1,\dots,\ell$, with at most finitely many potential future transitions 
occurring after $t_\ell$. If the latter set is empty, then the execution
of the circuit has already been determined completely. 
Otherwise, let $t_{\ell+1}>t_\ell$ 
be the closest future transition time. 

During iteration $\ell+1$, all transitions occurring
at time $t_{\ell+1}$ are dealt with, exactly as in the base case:
Any transition $e$, with causal depth $d(e)$, 
happening at $t_{\ell+1}$ at a gate output or at some
input port may cause a mode switch of every gate $G$ that 
is connected to it. Due to \cref{lem:ioseparationtime},
such a mode switch, and hence each of the at most finitely
many output transitions $e'$ occurring during 
that new mode (which all are assigned a causal depth $d(e')=d(e)+1$), 
of $G$ can only happen at or after time $t_{\ell+1}+\dmin^C$. 
In addition, the at most finitely many potential future transitions w.r.t.
$t_{\ell}$ of all gates that have not been mode-switched and actually
occur at times greater than $t_{\ell+1}$ are retained, along with their
assigned causal depth, as potential future
transitions w.r.t. $t_{\ell+1}$. Overall, we again end up with at most 
finitely many potential future transitions, which completes the induction step.

To complete our proof, we only need to argue that 
$\lim_{\ell\to\infty} t_\ell = \infty$ for the case where the iterations
do not stop at some finite $\ell$. This follows immediately from the fact that, for every
$k\geq 1$, there must be some iteration $\ell \geq 1$ such that 
$t_{\ell} \geq k\dmin^C$. If this was not the case, there must be
some iteration after which no further mode switch of any gate takes place.
This would cause the set of potential future transitions to shrink in
every subsequent iteration, however, and thus the simulation algorithm to stop,
which provides the required contradiction.
\end{proof}

From the execution construction, we also immediately get:

\begin{lem}\label{lem:Depth_Iteration}
For all~$\ell\ge 1$, (a) the simulation algorithm never assigns a
     causal depth larger than~$\ell$ to a transition generated in
     iteration~$\ell$, and (b) at the end of iteration~$\ell$ the sequence
     of causal depths of transitions in~$s_{\Gamma}$ for $t\in[0,t_\ell]$ is
     nondecreasing for all components~$\Gamma$.
\end{lem}

\subsection{Impossibility of short-pulse filtration}

The results of the previous subsection allow us to adapt the impossibility proof
of \cite{FNNS19:TCAD} to our setting. We start with the the definition of the
SPF problem:

\textbf{Short-Pulse Filtration.} 
A signal {\em contains a pulse\/} of length~$\Delta$ at time~$T_0$, if it
 contains a rising transition at time~$T_0$, a falling transition at time
 $T_0+\Delta$, and no transition in between.
The \emph{zero signal} has the initial value 0 and does not contain any transition.
A circuit {\em solves Short-Pulse Filtration (SPF)\/}, if it fulfills all of:
\begin{enumerate}
\item[F1)] The circuit has exactly one input port and exactly one output port. {\em
        (Well-formedness)}
        \label{def:SPF}
\item[F2)] If the input signal is the zero signal, then so is the output
	signal. {\em (No generation)}
\item[F3)] There exists an input pulse such that
	the output signal is not the zero signal. {\em (Nontriviality)}
\item[F4)] There exists an~$\varepsilon>0$ such that for every input pulse the output
	signal never contains a pulse of length less than or equal to~$\varepsilon$. {\em
	(No short pulses)}
\end{enumerate}
We allow the circuit to behave arbitrarily if the input
  signal is not a single pulse or the zero signal.

A circuit {\em solves bounded SPF\/} if additionally:
\begin{enumerate}
\item[F5)] There exists a $K>0$ such that for every input pulse
  the last output transition
  is before time~$T_0+\Delta+K$, where~$T_0$ is the time of the first input transition.
  {\em (Bounded stabilization time)}
\end{enumerate}

A circuit is called a {\em forward circuit\/} if its graph is acyclic.
Forward circuits are exactly those circuits that do not contain
     feedback loops.
Equipped with the continuity of digitized hybrid gates and the fact that
     the composition of continuous functions is continuous, it is not
     too difficult to prove that the inherently discontinuous SPF problem
     cannot be solved with forward circuits.

\begin{theorem}\label{thm:no_forward_circuit}
No forward circuit solves bounded SPF.
\end{theorem}
\begin{proof}
Suppose that there exists a forward circuit that solves bounded SPF with
stabilization time bound~$K$.
Denote by~$s_\Delta$ its output signal when feeding it a $\Delta$-pulse at
time~$0$ as the input.
Because~$s_\Delta$ in forward circuits is a finite composition of continuous
functions by
Theorem~\ref{thm:digitizedmodels:are:cont}, $\lVert s_\Delta \rVert_{[0,T],1}$
depends continuously on~$\Delta$, for any $T$.

By the nontriviality condition (F3) of the SPF problem, there exists
some~$\Delta_0$ such that $s_{\Delta_0}$ is not the zero signal.
Set $T = 2\Delta_0 + K$.

Let~$\varepsilon>0$ be smaller than both~$\Delta_0$ and $\lVert s_{\Delta_0} \rVert_{[0,T],1}$.
We show a contradiction by finding some~$\Delta$ such that~$s_\Delta$ either
contains a pulse of length less than~$\varepsilon$ (contradiction to the no
short pulses condition (F4)) or contains a transition
after time $\Delta+K$ (contradicting the bounded stabilization time
condition~(F5)).

Since $\lVert s_\Delta \rVert_{[0,T],1}\to0$ as $\Delta\to0$ by the no generation condition (F2)
of SPF,
there exists a~$\Delta_1<\Delta_0$ such that $\lVert s_{\Delta_1} \rVert_{[0,T],1}=\varepsilon$
by the intermediate value property of continuity.
By the bounded stabilization time condition (F5), there are no transitions
in~$s_{\Delta_1}$ after time $\Delta_1+K$.
Hence, $s_{\Delta_1}$ is~$0$ after this time because otherwise it is~$1$ for the
remaining duration $T - (\Delta_1+K) > \Delta_0 > \varepsilon$, which would
mean that $\lVert s_{\Delta_1} \rVert_{[0,T],1}>\varepsilon$.
Consequently, there exists a pulse in~$s_{\Delta_1}$ before
time $\Delta_1+K$.
But any such pulse is of length at most~$\varepsilon$ because
$\lVert s_{\Delta_1} \rVert_{[0,\Delta_1+K],1}  \leq 
\lVert s_{\Delta_1} \rVert_{[0,T],1}=\varepsilon$.
This is a contradiction to the no short pulses condition (F4).
\end{proof}

We next show how to simulate (part of) an execution of an arbitrary
     circuit~$\C$ by a forward circuit~$\C'$ generated from~$\C$ by
     the unrolling of feedback loops.
Intuitively, the deeper the unrolling, the longer the time~$\C'$
     behaves as~$\C$.

\begin{dfn}
Let~$\C$ be a circuit, $V$ a vertex of~$\C$, and $k\geq0$.
We define the {\em $k$-unrolling of~$\C$ from~$V$}, denoted by~$\C_k(V)$, 
to be a directed acyclic graph with a single sink, constructed
as follows:

The unrolling $\C_k(I)$ from input port $I$ is just a copy of that input port.
The unrolling $\C_k(O)$ from output port $O$ with incoming channel $C$ and predecessor $V$ comprises a copy of the output port $O^{(k)}$ and the unrolled circuit $\C_k(V)$ with its sink connected to $O^{(k)}$ by an edge. 

The 0-unrolling $\C_0(B)$ from hybrid gate $B$ is a trivial Boolean gate $X_v$ without inputs and the constant output value~$v$ equal to $B$'s initial digitized output value. For $k>0$, the $k$-unrolling $\C_k(B)$ from gate $B$ comprises an exact copy of that gate $B^{(k)}$. 
Additionally, for every incoming edge of $B$ from $V$ in $\C$, it contains the circuit $\C_{k-1}(V)$  with its sink connected to $B^{(k)}$.
Note that all copies of the same input port are considered to be the same.
\end{dfn}

To each component~$\Gamma$ in~$\C_k(V)$, we assign a value~$z(\Gamma)
\in\IN_0\cup\{\infty\}$ as follows:
     $z(\Gamma) = \infty$
     if~$\Gamma$ has no predecessor (in particular, is an input port) 
and $\Gamma\not\in\{X_0,X_1\}$. Moreover, $z(X_0)=z(X_1) = 0$, 
     $z(V)=z(U)$ if $V$ is an output port connected by an edge to $U$, and
     $z(B) = \min_{c\in E^B}\{1+z(c)\}$ if $B$ is a gate with its inputs connected
to the components in the set $E^B$.
\cref{fig:unrolling} shows an example of a circuit and an
     unrolled circuit with its $z$~values.

\begin{figure}
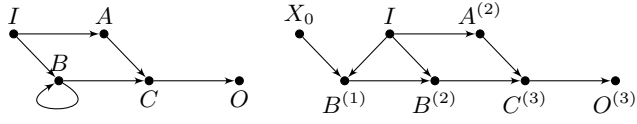

\centering
  {
    \tikzfigureunrolling}
\caption{Circuit~$\C$ (left) and $\C_3(O)$ (right) under the assumption that the gate~$B$ has initial value~$0$. It is $z(X_0)=0$,
	$z(I)=z(A^{(2)})=\infty$, $z(B^{(1)})=1$, $z(B^{(2)})=2$, $z(C^{(3)})=3$, and
$z(O^{(3)})=3$.}
\label{fig:unrolling}
\end{figure}

Noting that, for every component $\Gamma$ in $C_k(V)$, $z(\Gamma)$ is the number of gates on the shortest path from an $X_v$ node to $\Gamma$, or $z(\Gamma)=\infty$ if no such path exists, we immediately get:

\begin{lem}\label{lem:unrol:kisz}
  The $z$-value assigned to the sink vertex $V^{(k)}$ of a $k$-unrolling $\C_k(V)$ of $\C$ from $V$ satisfies $z(V^{(k)})\geq k$.
\end{lem}

Recalling the causal depths assigned to transitions during the execution construction in \cref{thm:execution}, we are now in the position to prove the result for a circuit simulated by an unrolled circuit.

\begin{theorem}\label{thm:simulation}
Let~$\C$ be a circuit with input port~$I$ and output port~$O$ that solves
bounded SPF.
Let~$\C_k(O)$ be an unrolling of~$\C$, $\Gamma$ a component in~$\C$, and
     $\Gamma'$ a copy of $\Gamma$ in~$\C_k(O)$.
For all input signals~$s_I$ on $I$, if a transition~$e$ is generated for~$\Gamma$
 by the execution construction algorithm run on
     circuit~$\C$ with input signal~$s_I$
     and~$d(e) \le z(\Gamma')$, then~$e$ is also generated
     for~${\Gamma'}$ by the algorithm run on circuit~$\C_k(O)$
     with input signal~$s_I$; and vice versa.
\end{theorem}
\begin{proof}
Assume that $e$ is the first transition violating the theorem.
The input signal is the same for both circuits, and the initial digitized values of gates in $\C$ and both their copies in $\C_k(O)$ and the $X_v$ gates resulting from their $0$-unrolling are equal as well. Hence, $e$ cannot be any such transition (added in iteration 1 only).

If $e$ was added to the output of a gate $B$ in either circuit, the transition $e'$ resp.\ $e''$ at one of its inputs that caused $e$ in $\C$ resp.\ $\C_k(V)$ must have been different. 
These transitions $e'$ resp.\ $e''$ must come from the output of some other gate $B_1$,
and causally precede $e$. Hence,
by \cref{def:inputoutputcausality}, $d(e)=d(e')+1$, and by \cref{lem:Depth_Iteration}, $d(e)\geq d(e'')$. Also by definition, $z(B)=z(B_1)+1$ in $C_k(O)$. Since $d(e)\leq z(B)$ by assumption, 
we find $d(e')\leq z(B_1)$ and $d(e'') \leq z(B)$, so applying our theorem to 
$e'$ and $e''$ yields a contradiction to $e$ being the first violating transition.
\end{proof}

We can finally prove that bounded SPF is not solvable, even with
non-forward circuits.

\begin{theorem}\label{thm:main_impossibility}
No circuit solves bounded SPF.
\end{theorem}

\begin{proof}
We first note that the impossibility of bounded SPF also implies the
impossibility of bounded SPF when restricting pulse lengths to be at most some
$\Delta_0>0$.

Since all transitions generated in the execution construction
\cref{thm:execution} up to any bounded time $t_\ell$ have bounded causal depth,
let~$\zeta$ be an upper bound on the causal depth of transitions up to the
SPF stabilization time bound~$\Delta_0+K$.
Then, by Theorem~\ref{thm:simulation} and Lemma~\ref{lem:unrol:kisz}, the $\zeta$-unrolled
circuit~$\C_\zeta(O)$ has the same output transitions
as the original circuit~$\C$ up to time~$\Delta_0+K$, and hence, by definition of
bounded SPF, the same transitions for all times.
But since~$\C_\zeta(O)$ is a forward circuit, it cannot solve bounded SPF by
Theorem~\ref{thm:no_forward_circuit}, i.e., neither can~$\C$.        
\end{proof}

\section{Digitized Hybrid Models for Multi-Input Gates}
\label{sec:MIG}
In this section, we will apply the results obtained in the previous
section to circuits composed of digitized hybrid gates. For a warm-up,
we will effortlessly re-prove the already known fact  
that every digitized hybrid gate model obtained by appending 
resp.\ prepending an IDM \emph{exp-channel} with pure delay
$\dmin>0$ at the output of resp.\ at every input of any zero-time Boolean gate
is continuous and strictly causal. 
Consequently, according to \cref{sec:circuits}, the resulting 
IDM circuit model is faithful w.r.t.\ solving the SPF problem.

An exp-channel, as introduced in 
\cite{FNNS19:TCAD}, is just the two-state digitized hybrid model illustrated
in \cref{fig:analogy} instantiated with exponential switching waveforms
$f_\downarrow(t)=1-f_\uparrow(t)=e^{-t/\tau}$ for some time constant 
$\tau >0$. Obviously, these are the trajectories of a simple first-order
RC low-pass filter. The ODEs governing $y=f_\downarrow(t)$ resp.\ 
$y=f_\uparrow$ are $y'+y/\tau=0$ resp.\ $y'+y/\tau=1/\tau$, so
$\Fdo(t,y)=-y/\tau$ resp.\ $\Fup(t,y)=(1-y)/\tau$ is of course
Lipschitz-continuous. An exp-channel hence satisfies the conditions 
of \cref{thm:channels:are:cont} and is hence continuous and, due to
the assumption $\dmin>0$, also strictly causal according to
\cref{def:strictcausality}. Since zero-time Boolean gates that alternate
with IDM channels can neither affect continuity nor causality of the latter, 
this completes our proof.

\subsection{Modeling multi-input switching effects}
As already mentioned in \cref{sec:intro}, experiments in \cite{OMFS20:INTEGRATION} 
showed that the prediction accuracy of the above IDM circuit model for multi-input
gates is below expectations. As revealed by Ferdowsi et al.~\cite{FMOS22:DATE}, this is primarily due to the fact that a model of a multi-input gate that combines
single-input single-output IDM channels with zero-time Boolean gates cannot 
properly capture output delay variations caused by \emph{multiple input switching} (MIS) effects: output transitions may be sped up/slowed down when different inputs switch in close temporal proximity~\cite{CGB01:DAC}. 

Consider the CMOS implementation of a \NOR\ gate shown in \cref{fig:nor_CMOS},
for example,
which consists of two serial \pmos\ ($T_1$ and $T_2$) for charging the load capacitance $C$ (producing a rising output transition), and two parallel \nmos\ transistors ($T_3$ and $T_4$) for discharging it (producing a falling one). 
When an input experiences a rising transition, the corresponding \nmos\ transistor closes while the corresponding \pmos\ transistor opens, so $C$ will be discharged.
If both inputs $A$ and $B$ experience a rising transition at the same time, 
$C$ is discharged twice as fast. Since the gate delay depends on the 
discharging speed, it follows that the delay 
$\dsd(\Delta)$ increases (by almost 30\% in the example shown in \cref{corFig3}) when the \emph{input separation time} 
$\Delta=t_B-t_A$ increases from 0 to $\infty$ or decreases from 0 to $-\infty$.
For falling input transitions, the behavior of the \NOR\ gate is quite different: \cref{corFig5} shows that the MIS effects lead to a moderate decrease of $\dsu(\Delta)$ when $|\Delta|$ goes from 0 to $\infty$, which is primarily caused
by capacitive coupling.

\begin{figure}[t!]
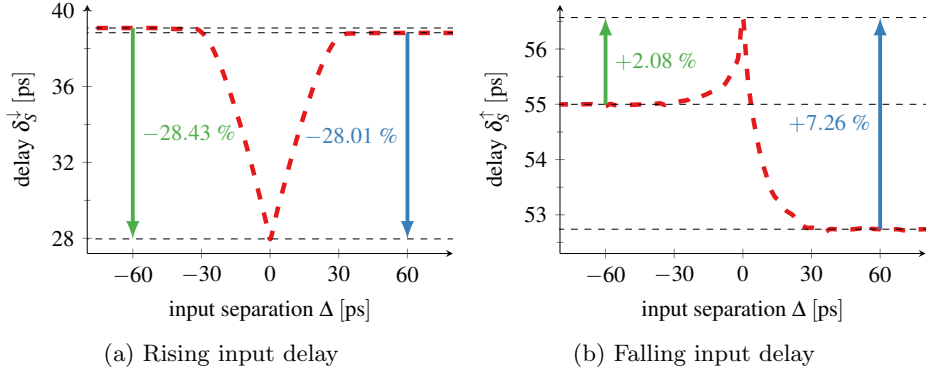

  \centering
  \subfloat[Rising input delay]{
    \includegraphics[width=0.48\linewidth]{\figPath{nor2_out_down_charlie_15nm_colored.pdf}}%
    \label{corFig3}}
  \hfil
  \subfloat[Falling input delay]{
    \includegraphics[width=0.48\linewidth]{\figPath{nor2_out_up_charlie_15nm_colored.pdf}}%
    \label{corFig5}}
  \caption{MIS effects in the measured delay of a $15$nm technology
CMOS \NOR\ gate.}\label{fig:Charlie15nmSim}
\end{figure}

MIS effects have of course been studied in the digital circuit modeling
literature in the past, with approaches ranging from linear~\cite{SRC15:TDAE} or quadratic~\cite{SKJPC09:ISOCC} fitting over higher-dimensional macromodels~\cite{CS96:DAC} and model representations \cite{SC06:DATE} to recent machine learning methods~\cite{RS21:TCAD}. However, the resulting models are either empirical or statistical and, hence, have not been analyzed w.r.t.\ continuity. Whether they
admit a faithful digital circuit model or not is hence unknown.

\subsection{A simple digitized hybrid model for a CMOS \NOR\ gate}
\label{sec:SimpleModel}
To the best of our knowledge, the first attempt to develop a delay model that captures MIS effects and can be analyzed w.r.t.\ continuity has been 
provided in \cite{FMOS22:DATE}. It is a 4-state digitized hybrid model for
a CMOS \NOR\ gate, with one mode per possible digital state of the inputs 
  $(A,B)\in \{(0,0), (0,1), (1,0), (1,1)\}$, which has been obtained by 
replacing the four transistors in \cref{fig:nor_CMOS} by ideal zero-time switches with non-zero resistance, 
and adding another capacitance $C_N$ to the node $N$ between the two \pmos\ transistors $T_1$ and $T_2$. 
In each mode, the voltage of the the output signal $O$ and the internal node $N$
are governed by a system of constant-coefficient first-order ODEs as follows: 

\begin{itemize}
\item System $(1,1)$: $\va=1$, $\vb=1$:
If inputs $A$ and $B$ are 1, both \nmos\ transistors are conducting and thus replaced by resistors, causing the output $O$ to be discharged in parallel. By contrast, $N$ is
completely isolated and keeps its value. This leads to the following ODEs:
\begin{align*}
  &\ \displaystyle{{\operatorname{d}\over\operatorname{d}\!t} \vint{}(t) \choose {\operatorname{d}\over\operatorname{d}\!t} \vout(t)}=
  \displaystyle{F_1(\vint{}(t), \vout(t)) \choose F_2(\vint{}(t), \vout(t))}=
  \displaystyle{0  \choose - \bigl(\frac{1}{\cout R_3} + \frac{1}{\cout R_4}\bigr) \vout(t)}
\end{align*}

\item System $(1,0)$: $\va=1$, $\vb=0$:
Since $T_1$ and $T_4$ are open, node $N$ is connected to $O$, and $O$ to \gnd. Both capacitors have to be discharged over resistor $R_3$, resulting in less current that is available for discharging~$\cout$. One obtains:
\begin{align*}
  &\ \displaystyle{{\operatorname{d}\over\operatorname{d}\!t} \vint{}(t) \choose {\operatorname{d}\over\operatorname{d}\!t} \vout(t)}=
  \displaystyle{F_3(\vint{}(t), \vout(t)) \choose F_4(\vint{}(t), \vout(t))}=
  \displaystyle{-\frac{\vint{}(t)}{\cint{} R_{2}}  +  \frac{\vout(t)}{\cint{} R_{2}}  \choose \frac{\vint{}(t)}{\cout R_{2}}  - \bigl(\frac{1}{\cout R_2} + \frac{1}{\cout 
                               R_3}\bigr) \vout(t) }
\end{align*}

\item System $(0,1)$: $\va=0$, $\vb=1$:
Opening transistors $T_2$ and $T_3$ again decouples the nodes $N$ and $O$. We thus get
\begin{align*}
  &\ \displaystyle{{\operatorname{d}\over\operatorname{d}\!t} \vint{}(t) \choose {\operatorname{d}\over\operatorname{d}\!t} \vout(t)}=
  \displaystyle{F_5(\vint{}(t), \vout(t)) \choose F_6(\vint{}(t), \vout(t))}=
  \displaystyle{-\frac{\vint{}(t)}{\cint{} R_{1}} +  \frac{\vdd}{\cint{} R_1}  \choose - \frac{\vout(t)}{\cout R_4} }
\end{align*}

\item System $(0,0)$: $\va=0$, $\vb=0$:
Closing both \pmos\ transistors causes both capacitors to be charged over the same resistor~$R_1$, similarly to system $(1,0)$. Thus
\begin{align*}
  &\ \displaystyle{{\operatorname{d}\over\operatorname{d}\!t} \vint{}(t) \choose {\operatorname{d}\over\operatorname{d}\!t} \vout(t)}=
  \displaystyle{F_7(\vint{}(t), \vout(t)) \choose F_8(\vint{}(t), \vout(t))}=\\
  &\ \displaystyle{-\bigl(\frac{1}{\cint{}(t) R_1} + \frac{1}{\cint{}(t) R_2}\bigr) \vint{} + \frac{\vout(t)}{\cint{} R_2} +  \frac{\vdd}{\cint{} R_1}  \choose \frac{\vint{}(t)}{\cout R_2}  - \frac{\vout(t)}{\cout R_2} }
\end{align*}
\end{itemize}

Every $F_i$, $i \in \{1, \ldots, 8 \}$, is a mapping from $U=(0,1)^2 \subseteq \IR^2$ to $\IR$, whereat $U$ is the vector of the voltages at the nodes $N$ and $O$. Solving the above ODEs provides analytic expressions for these voltage trajectories, which can even be inverted to obtain the relevant gate delays.
As it turned out in \cite{FMOS22:DATE}, although the model perfectly covers the MIS effects in the case of rising input transitions (\cref{Fig.Resup_simple}), it unfortunately fails to do so for falling input (= rising output) transitions (\cref{Fig.Resdown_simple}).

\begin{figure}[t!]
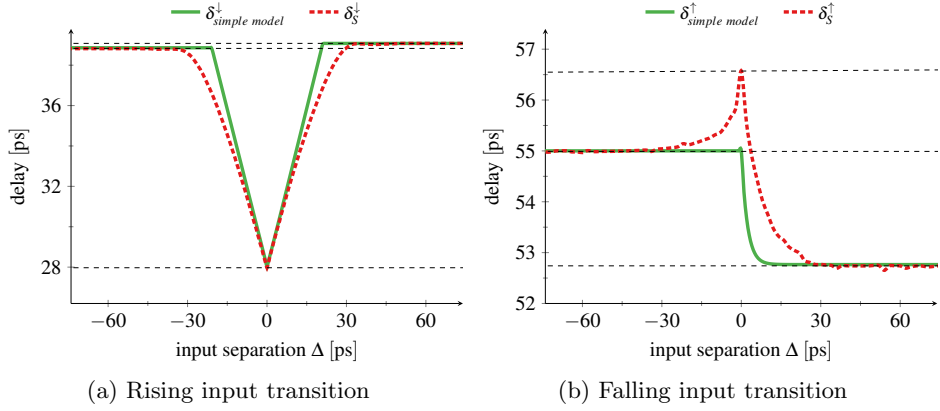

  \centering
  \subfloat[Rising input transition]{
    \includegraphics[width=0.485\linewidth]{\figPath{hm_falling_output_new_simplemodel.pdf}}%
    \label{Fig.Resup_simple}}
  \hfil
  \subfloat[Falling input transition]{
    \includegraphics[width=0.485\linewidth]{\figPath{rising_output_simplemodel.pdf}}%
    \label{Fig.Resdown_simple}}
  \caption{Comparison of the measured delay $\delta_S^{\downarrow/\uparrow}(\Delta)$ of a real $15$nm CMOS \NOR\ gate (red dashed line) and the delay prediction of the simple digitized hybrid model (green line) from \cite{FMOS22:DATE}.}\label{fig:CharlieResult_simple}
\end{figure}

Nevertheless, using \cref{thm:digitizedmodels:are:cont}, we can show that the 
model of \cite{FMOS22:DATE} is continuous:

\begin{theorem}[]\label{thm:SimpleModel}
  For any $i \in \{1, \ldots, 8 \}$, $F_i$ of the simple digitized hybrid model
  is Lipschitz continuous.
\end{theorem}
\begin{proof}
Since the proof is straighforward, we elaborate it only for $F_7$; similar arguments apply to the other cases. Let $K= max \bigl\{ (\frac{1}{\cint R_1} + \frac{1}{\cint R_2}), \frac{1}{\cint R_2} \bigr\}$.
For any voltages $\vint^{1}$, $\vint^{2}$, $\vout^{1}$, and $\vout^{2}$ in $(0,1)$, we find
\begin{flalign}
\bigl| F_7(\vint^{1}, \vout^{1}) - F_7(\vint^{2}, \vout^{2}) \bigr| &=
\Bigl| -\bigl(\frac{1}{\cint R_1} + \frac{1}{\cint R_2}\bigr)(\vint^{1}- \vint^{2})\\
&\qquad +
 \quad\frac{1}{\cint R_2}\bigl(\vout^{1}- \vout^{2}) \Bigr|  \\
&\leq K \bigl| (\vint^{1}- \vint^{2}) + (\vout^{1}- \vout^{2}) \bigr|.
\end{flalign}
\end{proof}

Consequently, we can instantiate \cref{def:dhm} with
\begin{align}
b_c(i_d^A,i_d^B) = \begin {cases}
\displaystyle{F_1(\vint(t), \vout(t)) \choose F_2(\vint(t), \vout(t))} &   \ \ (i_d^A,i_d^B)=(1,1)\\
\displaystyle{F_3(\vint(t), \vout(t)) \choose F_4(\vint(t), \vout(t))} &   \ \ (i_d^A,i_d^B)=(1,0)\\
\displaystyle{F_5(\vint(t), \vout(t)) \choose F_6(\vint(t), \vout(t))} &   \ \ (i_d^A,i_d^B)=(0,1)\\
\displaystyle{F_7(\vint(t), \vout(t)) \choose F_8(\vint(t), \vout(t))} &   \ \ (i_d^A,i_d^B)=(0,0) \nonumber
\end {cases}
\end{align}
such that the model is continuous by \cref{thm:digitizedmodels:are:cont}.

\section{An Advanced Digitized Hybrid Model for a CMOS \NOR\ Gate}
\label{sec:AdvancedModel}

In an attempt to mitigate the inability of the simple digitized hybrid
model for a CMOR \NOR\ gate proposed in \cite{FMOS22:DATE} to cover the
MIS effect for falling input (= rising output) transitions (recall \cref{Fig.Resdown_simple}),
Ferdowsi, Schmid, and Salzmann developed an advanced model originally
presented in \cite{ferdowsi2023accurate}. Whereas 
this model indeed accomplishes its purpose, its analysis is based
on a complicated piecewise approximation (in terms of $\Delta$) of both the
ODE solutions and, in particular, the corresponding delay formulas. 
This not only impairs the utility of the results for determining delays
of compound circuits, both for simulation-based and analytical studies, 
but also caused the model parametrization, which is based on fitting, 
to partially compensate for the approximation error by obtaining inexact parameters.

In this section, we will provide an entirely novel analysis of the
digitized hybrid model proposed in \cite{ferdowsi2023accurate},
which has been enabled by the recent discovery of an explicit expression
for the ODE solution. It not only leads to more accurate delay
formulas, but also to an explicit model parametrization procedure
that avoids any fitting.


The advanced digitized  hybrid model for a 2-input CMOS \NOR\ gate
introduced in \cite{ferdowsi2023accurate} is built upon replacing the
transistors in \cref{fig:nor_CMOS} by time-varying resistors:
The values $R_i(t)$, $i \in \{1,\ldots,4 \}$ in the resulting 
\cref{FigureNOR-GATE}
vary between some fixed on-resistance $R_i$ and the off-resistance
$\infty$ according to some laws, which we will introduce below.
The law to be used is determined by the state of the particular
input signal that drives the gate of the corresponding transitor.
This construction results in a hybrid model with 4 different modes, 
which correspond to the 4
possible input states $(A,B)\in \{(0,0), (0,1), (1,0), (1,1)\}$.

\begin{figure}[h]
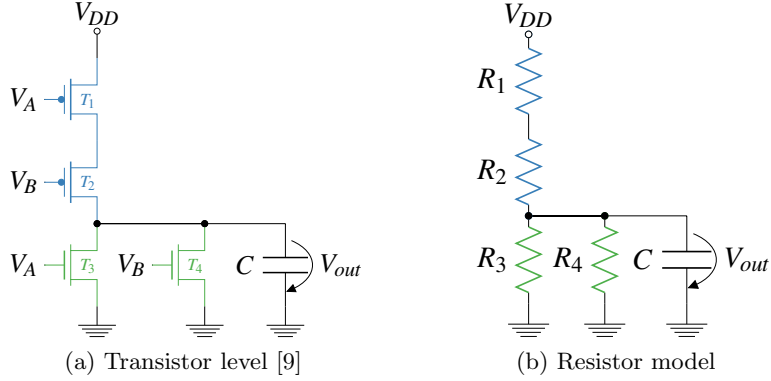

  \centering
  \subfloat[Transistor level~\cite{FMOS22:DATE}]{
\includegraphics[height=0.36\linewidth]{\figPath{nor_RC_colored.pdf}}  

    \label{fig:nor_CMOS}}
  \hfil
  \subfloat[Resistor model]{
 \includegraphics[height=0.36\linewidth]{\figPath{nor_R_colored.pdf}}%
    \label{FigureNOR-GATE}}
  \caption{Schematics and resistor model of a CMOS \NOR\ gate.}
\end{figure}

\cref{tab:T1} shows all possible input state transitions and the corresponding
resistor time evolution mode switches. Double arrows in the mode switch names
indicate MIS-relevant modes, whereas $+$ and $-$ indicate whether input $A$ switched
before $B$ or the other way round.  For instance, assume the system is in state
$(0,0)$ initially, i.e., that both $A$ and $B$ were set to 0 at time
$t_A= t_B= -\infty$. This causes $R_1$ and $R_2$ to be in the \emph{on-mode},
whereas $R_3$ and $R_4$ are in the \emph{off-mode}. Now assume that at time
$t_A=0$, $A$ is switched to 1. This switches $R_1$ resp.\ $R_3$ to the
\emph{off-mode} resp.\ \emph{on-mode} at time
$t^{\off}_1 = t^{\on}_3 = t_A = 0$. The corresponding mode switch is
$T_{-}^{\uparrow}$ and reaches state $(1,0)$. Now assume that $B$ is also
switched to 1, at some time $t_B=\Delta>0$.
This causes $R_2$ resp.\ $R_4$ to switch to \emph{off-mode} resp.\
\emph{on-mode} at time $t^{\off}_2= t^{\on}_4 = t_B=\Delta$. The corresponding
mode switch is $T_{+}^{\uparrow\uparrow}$ and reaches state $(1,1)$; note
carefully that the delay is $\Delta$-dependent and hence MIS-relevant.

\begin{table}[t]
\centering
\caption{State transitions and modes. $\uparrow$ and $\uparrow \uparrow$ (resp.\ $\downarrow$ and $\downarrow \downarrow$) represent the first and the second rising (resp.\ falling) input transitions. $+$ and $-$ specify the sign of the switching time difference $\Delta=t_B-t_A$.}
\scalebox{0.77}
{
\begin{tabular}{lllllllllll}
\hline
Mode                            &  & Transition                &  & $t_A$       & $t_B$       &  & $R_1$                & $R_2$                & $R_3$                & $R_4$                \\ \cline{1-1} \cline{3-3} \cline{5-6} \cline{8-11} 
$T^{\uparrow}_{-}$              &  & $(0,0) \rightarrow (1,0)$ &  & $0$         & $-\infty$   &  & $on \rightarrow off$ & $on$                 & $off \rightarrow on$ & $off$                \\
$T^{\uparrow \uparrow}_{+}$     &  & $(1,0) \rightarrow (1,1)$ &  & $-|\Delta|$ & $0$         &  & $off$                & $on \rightarrow off$ & $on$                 & $off \rightarrow on$ \\
$T^{\uparrow}_{+}$              &  & $(0,0) \rightarrow (0,1)$ &  & $-\infty$   & $0$         &  & $on$                 & $on \rightarrow off$ & $off$                & $off \rightarrow on$ \\
$T^{\uparrow \uparrow}_{-}$     &  & $(0,1) \rightarrow (1,1)$ &  & $0$         & $-|\Delta|$ &  & $on \rightarrow off$ & $off$                & $off \rightarrow on$ & $on$                 \\
$T^{\downarrow}_{-}$            &  & $(1,1) \rightarrow (0,1)$ &  & $0$         & $-\infty$   &  & $off \rightarrow on$ & $off$                & $on \rightarrow off$ & $on$                 \\
$T^{\downarrow \downarrow}_{+}$ &  & $(0,1) \rightarrow (0,0)$ &  & $-|\Delta|$ & $0$         &  & $on$                 & $off \rightarrow on$ & $off$                & $on \rightarrow off$ \\
$T^{\downarrow}_{+}$            &  & $(1,1) \rightarrow (1,0)$ &  & $-\infty$   & $0$         &  & $off$                & $off \rightarrow on$ & $on$                 & $on \rightarrow off$ \\
$T^{\downarrow \downarrow}_{-}$ &  & $(1,0) \rightarrow (0,0)$ &  & $0$         & $-|\Delta|$ &  & $off \rightarrow on$ & $on$                 & $on \rightarrow off$ & $off$                \\ \hline
\end{tabular}}
\label{tab:T1}
\end{table}

Crucial for the model is choosing a suitable law for the time evolution of 
$R_i(t)$ in the
\emph{on-} and \emph{off-mode}, which should facilitate an analytic solution of
the resulting ODE systems \cref{Eq:ODE_base} while being reasonably 
close to the physical behavior of a transistor. The simple Shichman-Hodges 
transistor model~\cite{ShichmanHodges} is used here,
which states a quadratic dependence 
of the output current on the input voltage. Approximating the latter by 
$d \sqrt{t-t_0}$ in the operation range close to the threshold voltage $\vth$, 
with $d$ and $t_0$ some fitting parameters, leads to the \emph{continuous resistance model}
\begin{align}
R_j^{\on}(t) &= \frac{\alpha_j}{t-t^{\on}}+R_j; \ t \geq t^{\on}, \label{on_mode}\\
R_j^{\off}(t) &= \beta_j (t-t^{\off}) +R_j; \ t \geq t^{\off}, \label{off_mode}
\end{align}
for some constant slope parameters $\alpha_j$ [\si{\ohm\s}], $\beta_j$
[\si[per-mode=symbol]{\ohm\per\s}], and on-resistance $R_j$ [\si{\ohm}].
$t^{\on}$ resp.\
$t^{\off}$ represent the time when the respective transistor is switched on
resp.\ off.

Actually, it was found in \cite{ferdowsi2023accurate} that
continuously changing resistors, according to \cref{on_mode},
are only required for switching-on the \pmos\ transistors in 
\cref{fig:nor_CMOS}. All other resistors can be immediately switched 
on/off (in zero-time), as already employed in \cite{FMOS22:DATE}. 
Note that immediate switching is obtained
by setting $\alpha_j=0$ and $\beta_j=\infty$ in \cref{on_mode} and \cref{off_mode}. Subsequently, we will use the notation $R_1=R_{p_A}$,
$R_2=R_{p_B}$ with the abbreviation $2R=R_{p_A}+R_{p_B}$ for the two \pmos\
transistors $T_1$ and $T_2$, and $R_3=R_{n_A}$, $R_4=R_{n_B}$ for the two \nmos\
transistors $T_3$ and $T_4$.

Another pivotal question is how to incorporate $R_1(t),\dots, R_4(t)$ in the ODEs of
the modes. The arguably most intuitive idea is to incorporate those in 
the state of the ODE of every mode, and switch between them continuously upon
a mode switch. This ``full-state model'' would lead to ODE systems with a 5-dimensional state (output voltage $\vout$ and the 4 resistors), however, which rules
out finding analytic solutions.

Therefore, in \cite{ferdowsi2023accurate}, these resistors were incorporated 
in the \emph{coefficients} of a simple first-order ODE 
obtained by applying Kirchhoff's rules to \cref{FigureNOR-GATE}. Doing this results
in the the non-autonomous, non-homogeneous ordinary differential equation (ODE) with non-constant coefficients
$ C\frac{\dd \vout}{\dd t} = \frac{\vdd-\vout}{R_1(t)+R_2(t)} -
\frac{\vout}{R_3(t)\ ||\ R_4(t)},$
which can be transformed into
\begin{equation}
\label{Eq:ODE_base}
\frac{\dd V_{out}}{\dd t} = F(t, V_{out}(t)) =   -\frac{V_{out}}{C\,R_g(t)}+U(t),
\end{equation}
where
$\frac{1}{R_g(t)}=\frac{1}{R_1(t)+R_2(t)}+\frac{1}{R_3(t)}+\frac{1}{R_4(t)}$ and
$U(t)=\frac{V_{DD}}{C(R_1(t)+R_2(t))}$. Note that the entire voltage divider in
\cref{FigureNOR-GATE} is equivalent to an ideal voltage source
$U_0=\vdd\frac{R_3(t)||R_4(t)}{R_1(t)+R_2(t)+R_3(t)||R_4(t)}$ and a serial
resistor $R_g(t)$ sourcing $C$. Consequently, $CU(t)=U_0/R_g(t)$ in \cref{Eq:ODE_base}
is the short-circuit current, and $CU(t)-\vout/R_g(t)$ the current actually
sourced into $C$.

\subsection{Continuity of the model}
In order to prove the continuity of the resulting digitized hybrid model, via \cref{thm:digitizedmodels:are:cont}, we need to verify
some properties of the functions $F(t,V_{out}(t))$ arising in the
ODE \cref{Eq:ODE_base}. Note carefully that, depending the current mode,
different expressions for $R_g(t)$, $U(t)$ determine the function
$F$ governing this mode. In fact, $F$ may even depend on
the actual mode switch, i.e., also the previous mode. 
\cref{T:Liptfunction} summarizes
the functions $F_1,\dots,F_6$ associated with each possible
input transition; unlike in \cref{tab:T1}, we also consider
state transitions where both inputs are changed simultaneously.
Due to some symmetry, we end up with only six different functions. 

For instance, to determine $F_5$ corresponding to the transition $(1,1) \rightarrow (0,0)$, we assume that the system is in mode $(1,1)$
initially (i.e., at time $t=-\infty$) and transitions to $(0,0)$ at
time $t=0$. Consequently, $R_1$ and $R_2$, previously in the \emph{off-mode}, switch to the \emph{on-mode}, while $R_3$ and $R_4$ switch from \emph{on-mode} to \emph{off-mode}. Formally, this transition results in $R_{p_{A}}(t)=\frac{\alpha_1}{t}+R_1$, $R_{p_{B}}(t)=\frac{\alpha_2}{t}+R_2$, and $R_{n_{A}}(t)=R_{n_{B}}(t)= \infty$, collectively leading to $1/ R_{g}(t)= t/(2Rt+\alpha_1+\alpha_2)$ since
$R_1+R_2=2R$. As a result, we obtain $F_5(t,\vout(t))=  -\frac{V_{out}}{C\,R_g(t)}+U(t) = \frac{(-V_{out}(t) +\vdd)t}{C ( 2Rt+ \alpha_1+\alpha_2 )}$. The other cases can be obtained similarly.

\begin{table}[t]
\centering
\caption{$F(t,V_{out}(t))$ for each state transition.} 
\scalebox{0.74}{
\begin{tabular}{lcl}
\hline
State transition          & \multicolumn{1}{l}{} & $F(t,V_{out}(t))$                                                                                                          \\ \cline{1-1} \cline{3-3} 
$(0,0) \rightarrow (1,0)$ &                      & $F_1(t,V_{out}(t)) \doteq \frac{-V_{out}(t)}{CR_{nA}}$                                                                     \\
$(1,1) \rightarrow (1,0)$ &                      & $F_1(t,V_{out}(t)) \doteq \frac{-V_{out}(t)}{CR_{nA}}$                                                                     \\
$(0,1) \rightarrow (1,0)$ &                      & $F_1(t,V_{out}(t)) \doteq \frac{-V_{out}(t)}{CR_{nA}}$                                                                     \\
$(0,0) \rightarrow (0,1)$ &                      & $F_2(t,V_{out}(t)) \doteq \frac{-V_{out}(t)}{CR_{nB}}$                                                                     \\
$(1,1) \rightarrow (0,1)$ &                      & $F_2(t,V_{out}(t)) \doteq \frac{-V_{out}(t)}{CR_{nB}}$                                                                     \\
$(1,0) \rightarrow (0,1)$ &                      & $F_2(t,V_{out}(t)) \doteq \frac{-V_{out}(t)}{CR_{nB}}$                                                                     \\
$(1,0) \rightarrow (0,0)$ &                      & $F_3(t,V_{out}(t)) \doteq \frac{ \bigl (-V_{out}(t) +\vdd \bigr )}{C(\frac{\alpha_1}{t}+\frac{\alpha_2}{t+ \Delta}+ 2R)}$  \\
$(0,1) \rightarrow (0,0)$ &                      & $F_4(t,V_{out}(t)) \doteq   \frac{\bigl (-V_{out}(t) +\vdd \bigr )}{C(\frac{\alpha_1}{t+ \Delta}+\frac{\alpha_2}{t}+ 2R)}$ \\
$(1,1) \rightarrow (0,0)$ &                      & $F_5 (t, V_{out}(t)) \doteq  \frac{\bigl (-V_{out}(t) +\vdd \bigr )t}{C ( 2Rt+ \alpha_1+\alpha_2 )}$                       \\
$(1,0) \rightarrow (1,1)$ &                      & $F_6(t,V_{out}(t)) \doteq \frac{-V_{out}(t)}{C}(\frac{1}{R_{n_A}}+\frac{1}{R_{n_B}})$                                      \\
$(0,1) \rightarrow (1,1)$ &                      & $F_6(t,V_{out}(t)) \doteq \frac{-V_{out}(t)}{C}(\frac{1}{R_{n_A}}+\frac{1}{R_{n_B}})$                                      \\
$(0,0) \rightarrow (1,1)$ &                      & $F_6(t,V_{out}(t)) \doteq \frac{-V_{out}(t)}{C}(\frac{1}{R_{n_A}}+\frac{1}{R_{n_B}})$                                      \\ \hline
\end{tabular}}
\label{T:Liptfunction}
\end{table}

The following theorem paves the way for verifying the continuity property of
the model, by guaranteeing the properties required in \cref{def:dhm}:

\begin{theorem}[]\label{thm:AdvancedModel}
  Let $F= \{F_1 , \ldots, F_6 : \IR \times (0,1) \rightarrow \IR \}$ be the set of all functions described in \cref{T:Liptfunction}. Every $F_i \in F$, where $i \in \{1, \ldots, 6 \}$, is continuous for $t \in [0,T]$, $0 < T < \infty$, $V_{out}\in (0,1)$, and Lipschitz continuous w.r.t.\ $V_{out}$.
\end{theorem}
\begin{proof}
The statement is immediate for functions $F_1$, $F_2$, and $F_6$. For
$F_5$, let $g(t) = \frac{t}{C(2Rt+ \alpha_1+ \alpha_2)}$. Since $t \in [0,T]$, $g(t)$ takes its supremum value in the interval, which we denote by $K$ (i.e., $sup_{t \in [0,T]} g(t)=K$). We observe
\begin{align*}
&\bigl| F_5(t, V_{out}^{1}) - F_5(t, V_{out}^{2}) \bigr| =  \Bigl| \frac{\bigl (-V_{out}^{1} +\vdd \bigr ) t}{C(2Rt+ \alpha_1+ \alpha_2)}- \frac{\bigl (-V_{out}^{2} +\vdd \bigr ) t}{C(2Rt+ \alpha_1+ \alpha_2)} \Bigr|\\
& =\Bigl| \frac{-t}{C(2Rt+ \alpha_1+ \alpha_2)} \cdot (V_{out}^{1}- V_{out}^{2}) \Bigr| \leq | K | \bigl| V_{out}^{1}- V_{out}^{2} \bigr|,
\end{align*}
which concludes the proof for $F_5$. The proof for $F_3$ and $F_4$ follows the same route; we only sketch the proof for $F_3$: We observe
\begin{align*}
& \bigl| F_3(t, V_{out}^{1}) - F_3(t, V_{out}^{2}) \bigr| = \Bigl| \frac{-(V_{out}^{1}- V_{out}^{2})}{\frac{\alpha_1}{t+ \Delta}+\frac{\alpha_2}{t}+ 2R} \Bigr|.
\end{align*}
Since we can safely assume that both $t$ and $t+\Delta$ belong to the closed interval $[0,T]$, by choosing $T$ appropriately, we obviously get some Lipschitz constant $L$ that
is independent of $t$. Consequently,
\begin{align*}
& \bigl| F_3(t, V_{out}^{1}) - F_3(t, V_{out}^{2}) \bigr| \leq  L \cdot \bigl|  \bigl ( V_{out}^{1}- V_{out}^{2} \bigr) \bigr|,
\end{align*}
which completes the proof.
\end{proof}

According to \cref{thm:AdvancedModel}, by defining $s(t)=(i_d^A(t^+),i_d^B(t^+))$ and $s_p(t)=(i_d^A(t),i_d^B(t))$, we can instantiate \cref{def:dhm} by the choice function

\begin{align}
b_c(s(t)) = \begin {cases}
F_1(t, V_{out}(t)) &   \ \ s(t)=(1,0)\\ 
F_2(t, V_{out}(t)) &  \ \  s(t)=(0,1) \\
F_3(t, V_{out}(t)) &  \ \ s(t)=(0,0), s_p(t)=(1,0)\\
F_4(t, V_{out}(t)) &  \ \ s(t)=(0,0), s_p(t)=(0,1)\\
F_5(t, V_{out}(t)) &  \ \ s(t)=(0,0), s_p(t)=(1,1)\\
F_6(t, V_{out}(t)) &  \ \ s(t)=(1,1) \nonumber
\end {cases}
\end{align}
which, according to \cref{Eq:ODE_base} and \cref{T:Liptfunction}, results in 
\begin{align}
\frac{d V_{out}(t)}{dt}=\begin {cases}
\frac{-V_{out}(t)}{CR_{nA}} &   \ \ s(t)=(1,0)\\ 
\frac{-V_{out}(t)}{CR_{nB}} &  \ \  s(t)=(0,1) \\
\frac{\bigl (-V_{out}(t) +\vdd \bigr ) t(t+\Delta)}{C \bigl( 2Rt^2+(\alpha_1+\alpha_2+2 \Delta R)t + \alpha_1 \Delta \bigr)} &  \ \ s(t)=(0,0), s_p(t)=(1,0)\\
\frac{\bigl (-V_{out}(t) +\vdd \bigr ) t(t+\Delta)}{C \bigl( 2Rt^2+(\alpha_1+\alpha_2+2 \Delta R)t + \alpha_2 \Delta \bigr)} &  \ \ s(t)=(0,0), s_p(t)=(0,1)\\
\frac{\bigl (-V_{out}(t) +\vdd \bigr ) t}{C ( 2Rt+ \alpha_1+\alpha_2  )}  &  \ \ s(t)=(0,0), s_p(t)=(1,1)\\
\frac{-V_{out}(t)}{C}(\frac{1}{R_{n_A}}+\frac{1}{R_{n_B}}) &  \ \ s(t)=(1,1). \nonumber
\end {cases}
\end{align}
Since all the conditions in \cref{def:dhm} are satisfied,
\cref{thm:digitizedmodels:are:cont} indeed
guarantees continuity of the model.

\subsection{Analytic solutions for the output voltage trajectories}
\label{Sec:Charlie}

We now turn our attention to the ability of our model to cover
all MIS effects illustrated in \cref{fig:Charlie15nmSim}.
Since gate delays are just the time it takes for the output voltage trajectory
to reach the threshold voltage, this subsection is devoted to
determining explicit analytic expressions for $\vout^{MS}(t)$
for each mode switch $MS$ listed in \cref{tab:T1}.

It is well-known that the general solution of \eqref{Eq:ODE_base} is
\begin{equation}
\label{Eq1}
    V_{out}(t)= V_0\ e^{-G(t)} + \int_{0}^{t} U(s)\ e^{G(s)-G(t)}\dd s,
\end{equation}
where $V_0=V_{out}(0)$ denotes the initial condition and $G(t) = \int_{0}^{t}
(C\,R_g(s))^{-1} \dd s$.

As already mentioned, $R_g(t)$ and $U(t)$ depend on the particular mode, recall
\cref{T:Liptfunction}. It turns out that computing $G(t)$ for each mode requires
the solution of three different integrals $I_1=\int_{0}^{t} \frac{\dd s}{R_1(s)+R_2(s)}$, $I_2=\int_{0}^{t}\frac{\dd s}{R_3(s)}$, and $I_3(t)=\int_{0}^{t} \frac{\dd s}{R_4(s)}$. \cref{T:InerInt} lists these integrals as well as the value $U(t)$ for each mode.

\begin{table}[h]
\centering
\caption{Integrals $I_1(t)$, $I_2(t)$, $I_3(t)$ and the function $U(t)$ for every possible mode switch; $\Delta=t_B-t_A$, and $2R=R_{p_A}+R_{p_B}$.}
\scalebox{0.73}
{
\begin{tabular}{llllll}
\hline
Mode                            &  & $I_1(t)= \int_{0}^{t} \frac{\dd s}{R_1(s)+R_2(s)}$                         & $I_2(t)= \int_{0}^{t}\frac{\dd s}{R_3(s)}$ & $I_3(t)=\int_{0}^{t} \frac{\dd s}{R_4(s)}$ & $U(t)= \frac{\vdd}{C(R_1(t)+R_2(t))}$                                                               \\ \cline{1-1} \cline{3-6} 
$T^{\uparrow}_{-}$              &  & $0$                                                                        & $\int_{0}^{t} (1/R_{n_A})\dd s$            & $0$                                        & $0$                                                                                                 \\
$T^{\uparrow \uparrow}_{+}$     &  & $0$                                                                        & $\int_{0}^{t} (1/R_{n_A})\dd s$            & $\int_{0}^{t} (1/R_{n_B}) \dd s$           & $0$                                                                                                 \\
$T^{\uparrow}_{+}$              &  & $0$                                                                        & $0$                                        & $\int_{0}^{t} (1/(R_{n_B}) \dd s$          & $0$                                                                                                 \\
$T^{\uparrow \uparrow}_{-}$     &  & $0$                                                                        & $\int_{0}^{t} (1/R_{n_A}) \dd s$           & $\int_{0}^{t} (1/R_{n_B}) \dd s$           & $0$                                                                                                 \\
$T^{\downarrow}_{-}$            &  & $0$                                                                        & $0$                                        & $\int_{0}^{t} (1/R_{n_B}) \dd s$           & $0$                                                                                                 \\
$T^{\downarrow \downarrow}_{+}$ &  & $\int_{0}^{t}(1/(\frac{\alpha_1}{s+\Delta}+\frac{\alpha_2}{s}+2R))\dd s$   & $0$                                        & $0$                                        & $\frac{\vdd t(t+ \Delta)}{C(2 R t^2 +(\alpha_1 + \alpha_2 + 2 \Delta R)t + \alpha_2 \Delta)}$       \\
$T^{\downarrow}_{+}$            &  & $0$                                                                        & $\int_{0}^{t} (1/(R_{n_A}) \dd s$          & $0$                                        & $0$                                                                                                 \\
$T^{\downarrow \downarrow}_{-}$ &  & $\int_{0}^{t}(1/(\frac{\alpha_1}{s}+\frac{\alpha_2}{s+|\Delta|}+2R))\dd s$ & $0$                                        & $0$                                        & $\frac{\vdd t(t+ |\Delta|)}{C(2 R t^2 +(\alpha_1 + \alpha_2 + 2 |\Delta| R)t + \alpha_1 |\Delta|)}$ \\ \hline
\end{tabular}}
\label{T:InerInt}
\end{table}

Fortunately, a closer look at \cref{tab:T1} and \cref{T:InerInt} shows a certain symmetry between the pairs of modes $(T^{\uparrow}_{-},T^{\uparrow}_{+})$, $(T^{\uparrow \uparrow}_{+},T^{\uparrow \uparrow}_{-})$, $(T^{\downarrow}_{-},T^{\downarrow}_{+})$, and $(T^{\downarrow \downarrow}_{+},T^{\downarrow \downarrow}_{-})$. Therefore, it is sufficient to derive analytic expressions for the case $\Delta\geq 0$ only. The corresponding formulas for $\Delta < 0$ can be obtained from those by exchanging  $\alpha_1$ and $\alpha_2$, as well as $R_{n_A}$ and $R_{n_B}$.


To proceed, we split the discussions into two parts, one devoted to rising input transitions (= falling output transitions) and one to falling input transitions (= rising output transitions). We start with the (simpler) former one.

\subsubsection{Rising input transitions}
In this part, we analyze the output voltage trajectories related to rising input transitions. The following theorem elaborates on this.

\begin{theorem}[Output trajectories for rising input transitions]\label{thm:RITTrajectory}
For any $0 \leq |\Delta| \leq \infty$, the voltage output trajectory functions of our model for rising input transitions are given by:
\begin{flalign}
V_{out}^{T^{\uparrow}_{-}}(t) &= V_{out}^{T^{\uparrow}_{-}}(0) e^{\frac{-t}{C R_{n_{A}}}}
\label{outsig1}\\
V_{out}^{T^{\uparrow}_{+}} (t) &= V_{out}^{T^{\uparrow}_{+}}(0) e^{\frac{-t}{C R_{n_{B}}}}
\label{outsig1neg}\\
V_{out}^{T^{\uparrow \uparrow}_{+}}(t) &=V_{out}^{T^{\uparrow}_{-}} (\Delta)  e^{- \bigl(\frac{1}{CR_{n_A}}+\frac{1}{CR_{n_B}}\bigr)t}
\label{outsig2}\\
V_{out}^{T^{\uparrow \uparrow}_{-}}(t) &=V_{out}^{T^{\uparrow}_{+}} (\Delta)  e^{- \bigl(\frac{1}{CR_{n_A}}+\frac{1}{CR_{n_B}}\bigr)t}
\label{outsig2_neg}
\end{flalign}
\end{theorem}
\begin{proof}
In order to compute $V_{out}^{T^{\uparrow}_{-}}(t)$, consider the corresponding integrals $I_1(t)$, $I_2(t)$, and $I_3(t)$, as well as $U(t)$ in the first line of \cref{T:InerInt}. Since we assumed immediate resistor switching here, we have $\beta_{1}=\beta_2 = \infty$ and $\alpha_3=\alpha_4=0$, so that
\begin{equation}
I_1(t)=I_3(t)=U(t)=0,\qquad
I_2(t)= \frac{t}{R_{n_A}}.\nonumber
\end{equation}
Since $G(t)=(I_1(t)+I_2(t)+I_3(t))/C$, we get $e^{\pm G(t)} = e^{\frac{\pm t}{CR_{n_A}}}$ and $\int_{0}^{t} e^{G(s)} U(s)ds =0$. With $V_0^{\uparrow} = V_{out}^{T^{\uparrow}_{-}}(0)$ as our initial value, \cref{Eq1} 
finally provides
\begin{align}
V_{out}^{T^{\uparrow}_{-}}(t) = V_{out}^{T^{\uparrow}_{-}}(0) e^{\frac{-t}{C R_{n{A}}}}. \nonumber
\end{align}
Similarly, for the mode $T^{\uparrow \uparrow}_{+}$, we obtain
\begin{equation}
I_1(t)=U(t)=0,\qquad I_2(t)=\frac{t}{R_{n_A}}, \qquad I_3(t)=\frac{t}{R_{n_B}},\nonumber
\end{equation}
such that $e^{\pm G(t)}= e^{\pm (\frac{1}{CR_{n_A}}+\frac{1}{CR_{n_B}})t}$ and $\int_{0}^{t} e^{G(s)} U(s) ds=0$. Consequently, we obtain
\begin{flalign}
&V_{out}^{T^{\uparrow \uparrow}_{+}}(t) =V_{out}^{T^{\uparrow}_{-}} (\Delta)  e^{- \bigl(\frac{1}{CR_{n_A}}+\frac{1}{CR_{n_B}}\bigr)t}, \nonumber
\end{flalign}
where the initial value $V_{out}^{T^{\uparrow}_{-}} (\Delta)$ can be computed via \cref{outsig1}. 

Due to our symmetry argument, exchanging $R_{n_A}$ and $R_{n_B}$
immediately provides the trajectories \cref{outsig2_neg} and \cref{outsig1neg} for negative $\Delta$. 
\end{proof}

\subsubsection{Falling input transitions}
In this case, we first need to compute $V_{out}^{T^{\downarrow}_{-}}(t)$. 
Again plugging the immediate switching parameters $\beta_{1}=\beta_2 = \infty$ and $\alpha_3=\alpha_4=0$ in the corresponding
expressions in \cref{T:InerInt} provides
$I_1(t)=I_2(t)=U(t)=0$ and $I_3(t) = \frac{t}{R_{n_B}}$.
With $V_{0}^{\downarrow}= V_{out}^{T_{-}^{\downarrow}}(0)$ as our initial condition, 
\cref{Eq1} yields 
\begin{equation}
V_{out}^{T^{\downarrow}_{-}}(t) = V_{out}^{T_{-}^{\downarrow}}(0) e^{\frac{-t}{CR_{n_B}}}.
\label{eq:FirstFall}
\end{equation}

Turning our attention to $V_{out}^{T^{\downarrow \downarrow}_{+}}(t)$ in
\cref{T:InerInt}, we are confronted
with a more intricate case: Whereas $I_2(t)=I_3(t)=0$ again, evaluating $I_1(t)$ requires us to study the function
$f(s)= \frac{1}{\frac{\alpha_1}{s+ \Delta}+ \frac{\alpha_2}{s}+2R}$,
as 
\begin{align}
\label{crucial_I1}
&I_1(t)=\int_{0}^{t} f(s) ds, \qquad G(t)=I_1(t)/C,\\
\label{crucial_expint}
&\int_{0}^{t} e^{G(s)} U(s) ds= \frac{\vdd}{C} \int_{0}^{t} e^{\frac{I_1(s)}{C}} f(s) ds.
\end{align}
It is not difficult to check that $\frac{\vdd}{C}\int_{0}^{t} e^{\frac{I_1(s)}{C}} f(s) ds=\vdd(e^{\frac{I_1(t)}{C}}-1)$, which according to \cref{Eq1} leads to
\begin{align}
\label{Vout_formula}
&V_{out}^{T^{\downarrow \downarrow}_{+}}(t)=(V_{out}^{T^{\downarrow}_{-}}(\Delta)-\vdd) e^{-I_1(t)/C}+ \vdd,
\end{align}
where $V_{out}^{T^{\downarrow}_{-}}(\Delta)$ gives the initial value.

In order to compute an explicit formula for the voltage trajectory from
\cref{Vout_formula}, we need to evaluate $e^{-I_1(t)/C}$. To simplify our
derivations, we write $f(s)=\frac{1-g(s)}{2R}$, where $g(s)=\frac{as+c'}{s^2+ds+c'}$ and
\begin{align}
  a&=\frac{\alpha_1+\alpha_2}{2R},\label{eq:a}\\
  d&=a+\Delta,\label{eq:d}\\
  c'&=\frac{\alpha_2 \Delta}{2R},\label{eq:cprime}\\
  \chi&=d^2-4c' = (a+\Delta)^2-\frac{2\alpha_2 \Delta}{R}.\label{eq:chi}
\end{align}  
  With this, $I_1(t)=\frac{1}{2R}(t-\int_{0}^{t}g(s) ds)$.
The following lemma reveals that  the denominator of $g$ possesses two rational zeros, which will make it easy to compute $\int_{0}^{t}g(s) ds$ after a simple
partial fraction decomposition.

\begin{lem}
  $s^2+ds+c'=0$ has two rational roots $s_1=\frac{-d + \sqrt{\chi}}{2}$
  and $s_2=\frac{-d - \sqrt{\chi}}{2}$, which satisfy
  $s_1s_2=c'$, $s_1+s_2=-d$, and $s_2-s_1=-\sqrt{\chi}$.
\label{lemma1}
\end{lem}
\begin{proof}
It is apparent that $s_{1,2}= \frac{-d \pm \sqrt{\chi}}{2}= \frac{a+\Delta}{2}(-1+ \sqrt{1- \frac{4b\Delta}{(a+\Delta)^2}})$, where $b=\alpha_2/(2R)$, are the two zeros of $s^2+ds+c'=0$. These zeroes are rational if and only if $1- \frac{4b\Delta}{(a+\Delta)^2} \geq 0$, i.e., if and only if 
\begin{align}
\label{quadterm}
&\Delta^2+(2a-4b)\Delta+a^2 \geq 0.
\end{align}
Clearly, \cref{quadterm} has two complex zeros $\Delta_{1,2}= \frac{(\alpha_2-\alpha_1) \pm \sqrt{ -4 \alpha_1 \alpha_2}}{2R}$ since $\alpha_1$, $\alpha_2$, and $R$ are all positive. Therefore, \cref{quadterm} cannot become negative for
any $\Delta$, since it is positive for $\Delta=0$. Consequently, $s_{1}$ and $s_2$ are rational, and satisfy $s_1s_2=c'$ and $s_1+s_2=-d$ by Vieta's theorem.
\end{proof}

The following theorem provides the sought explicit expression for
$I_1(t)=\frac{1}{2R} \bigl( t-\int_{0}^{t}g(s) ds \bigr)$:
\begin{lem}
  Let $s_1$ and $s_2$ denote the two rational zeros of $s^2+ds+c'=0$, and
  define 
\begin{equation}
A=\frac{-as_1-c'}{s_2-s_1} = \frac{a\frac{d-\sqrt{\chi}}{2}-\frac{\alpha_2\Delta}{2R}}{-\sqrt{\chi}}.\label{eq:A}
\end{equation}
  Then,
\begin{equation}
  I_1(t)= \frac{1}{2R} \left[ t+ (A-a) \cdot \log\Bigl(1+\frac{2t}{d+\sqrt{\chi}}\Bigr) - A \cdot \log\Bigl(1+\frac{2t}{d-\sqrt{\chi}}\Bigr) \right].
\end{equation}
\end{lem}
\begin{proof}
  Utilizing partial fraction decomposition and recalling $s_1s_2=c'$ and $s_1+s_2=-d$ from \cref{lemma1}
  gives us $g(s)=\frac{as+c'}{s^2+ds+c'}=\frac{A}{(s-s_1)}+ \frac{a-A}{(s-s_2)}$,
for $A$ as defined in \cref{eq:A},
 which leads to
\begin{flalign}
\int g(s)ds &= A \int \frac{ds}{s-s_1}+ (a-A) \int \frac{ds}{s-s_2} \nonumber \\
&  =  A \cdot \log(s-s_1) + (a-A) \cdot \log(s-s_2)+K \nonumber \\
&  =  A \cdot \log\Bigl(s-\frac{-d + \sqrt{\chi}}{2}\Bigr) + (a-A) \cdot \log\Bigl(s-\frac{-d - \sqrt{\chi}}{2}\Bigr)+K \nonumber\\
&  =  (a-A) \cdot \log\Bigl(s+\frac{d + \sqrt{\chi}}{2}\Bigr) +  A \cdot \log\Bigl(s+\frac{d - \sqrt{\chi}}{2}\Bigr)+K \nonumber,
\end{flalign}
where $K$ is some constant. Plugging in the boundaries,
elementary calculations finally yield
\begin{flalign}
& I_1(t)=\int_{0}^{t} f(s)ds = \int_{0}^{t} \frac{1-g(s)}{2R} ds= \nonumber \\
&\frac{1}{2R} \left[ t+ (A-a) \cdot \log\Bigl(1+\frac{2t}{d+\sqrt{\chi}}\Bigr) - A \cdot \log\Bigl(1+\frac{2t}{d-\sqrt{\chi}}\Bigr) \right]
\label{TheoremGTerm}
\end{flalign}
as asserted.
\end{proof}

With these preparations, we are now ready to state the major theorem of this subsection:

\begin{theorem}[Output trajectories for falling input transitions]\label{thm:FITTrajectory}
For any $0 \leq |\Delta| \leq \infty$, the voltage output trajectory functions of our model for input falling transitions are given by:
\begin{flalign}
V_{out}^{T^{\downarrow}_{-}}(t) &= V_{out}^{T_{-}^{\downarrow}}(0) e^{\frac{-t}{CR_{n_B}}}
\label{eq:FirstFalltheorem}\\
V_{out}^{T^{\downarrow}_{+}}(t) &= V_{out}^{T_{+}^{\downarrow}}(0) e^{\frac{-t}{CR_{n_A}}}
\label{eq:FirstFalltheorem_plus}\\
V_{out}^{T^{\downarrow \downarrow}_{+}}(t)&= \vdd \label{SoughtOutput} \\
&\qquad + \bigl(V_{out}^{T^{\downarrow}_{-}}(\Delta) -\vdd   \bigr) \left[ e^{\frac{-t}{2RC}} \Bigl(1+\frac{2t}{d+\sqrt{\chi}}\Bigr)^{\frac{-A+a}{2RC}} \Bigl(1+\frac{2t}{d-\sqrt{\chi}}\Bigr)^{\frac{A}{2RC}} \right]\nonumber\\
V_{out}^{T^{\downarrow \downarrow}_{-}}(t)&= \vdd \label{SoughtOutput_neg} \\
&\qquad + \bigl(V_{out}^{T^{\downarrow}_{+}}(|\Delta|) -\vdd   \bigr) \left[ e^{\frac{-t}{2RC}} \Bigl(1+\frac{2t}{d+\sqrt{\chi}}\Bigr)^{\frac{-A+a}{2RC}} \Bigl(1+\frac{2t}{d-\sqrt{\chi}}\Bigr)^{\frac{A}{2RC}} \right]\nonumber
\end{flalign}
\end{theorem}
\begin{proof}
The trajectory \cref{eq:FirstFalltheorem} has been established in \cref{eq:FirstFall} already,
\cref{eq:FirstFalltheorem_plus} follows from our symmetry argument by exchanging $R_{n_B}$
with $R_{n_A}$.

Plugging in $I_1(t)$ in \cref{TheoremGTerm} into \cref{Vout_formula}, we immediately obtain the 
expression for  the output trajectory $V_{out}^{T^{\downarrow \downarrow}_{+}}(t)$
starting from the initial value $V_{out}^{T^{\downarrow}_{-}}(\Delta)$ given in \cref{eq:FirstFall}.
Due to our symmetry, the trajectory formula \cref{SoughtOutput_neg}
for negative values of $\Delta$ is obtained by exchanging $\alpha_1$ with
$\alpha_2$ in $\chi$ and $A$, and $V_{out}^{T^{\downarrow}_{-}}(\Delta)$ with $V_{out}^{T^{\downarrow}_{+}}(|\Delta|)$ in \cref{SoughtOutput}.
\end{proof}

\subsection{Delay formulas}

With the explicit output trajectories available, we can now
determine formulas for the MIS gate delays of our model, which
are functions of the input separation time $\Delta=t_B-t_A$.
We use the following general procedure, which we exemplify for
the case $\Delta \geq 0$; the case $\Delta < 0$ follows
by invoking our symmetry argument again.

\begin{itemize}
\item For rising input transitions (= falling output transitions), we compute $\vout^{T^{\uparrow}_{-}}(\Delta)$,
  and use it as the initial value for $\vout^{T^{\uparrow \uparrow}_{+}}(t)$. The sought MIS gate delay $\delta_{M,+}^{\downarrow}(\Delta)$ is the time until the latter crosses the threshold voltage $\vdd/2$.
\item For falling input transitions (= rising output transitions), we compute $\vout^{T^{\downarrow}_{-}}(\Delta)$, and use it as the initial value for $\vout^{T^{\downarrow \downarrow}_{+}}(t)$. The sought MIS gate delay $\delta_{M,+}^{\uparrow}(\Delta)$ is the time until the latter crosses the threshold voltage $\vdd/2$.
\end{itemize}

\subsubsection{Rising input transitions}

We again start with the simpler rising input transition scenario:

\begin{theorem}[MIS delay functions for rising input transitions]\label{thm:delayfunctions}
  For any $0 \leq |\Delta| \leq \infty$, the MIS gate delay functions of our model
  for rising input transitions  are given by:
\begin{align}
\delta_{M,+}^{\downarrow}(\Delta) = \begin {cases}
 \frac{\log(2)CR_{n_A}R_{n_B} - \Delta R_{n_B}}{R_{n_A}+R_{n_B}} + \Delta &   \ \ 0 \leq \Delta < \log(2)CR_{n_A} \nonumber \\ 
 \log(2)CR_{n_A} &   \ \ \Delta \geq \log(2)CR_{n_A}
\end {cases}
\end{align}
\begin{align}
\delta_{M,-}^{\downarrow}(\Delta) = \begin{cases}
 \frac{\log(2)CR_{n_A}R_{n_B} + |\Delta| R_{n_A}}{R_{n_A}+R_{n_B}} + |\Delta| &   \ \ |\Delta| < \log(2)CR_{n_B} \nonumber \\ 
\log(2)CR_{n_B} &   \ \ |\Delta| \geq \log(2)CR_{n_B}
\end {cases}
\end{align}
\end{theorem}

\begin{proof}
We sketch how $\delta_{M,+}^{\downarrow}(\Delta)$ is computed; the expression for $\delta_{M,-}^{\downarrow}(\Delta)$ is obtained analogously by our usual symmetry argument. Consider the trajectory $V_{out}^{T^{\uparrow \uparrow}_{+}}(t)$ in \cref{outsig2} starting from the initial value $V_{out}^{T^{\uparrow}_{-}}(\Delta)$, where the latter
in turn is started from the initial value $V_{out}^{T^{\uparrow}_{-}}(0)=\vdd$. The objective is to compute the time $\delta_{M,+}^\downarrow(\Delta)$ when $\vdd/2$ is hit by either (i) already the preceding trajectory $V_{out}^{T^{\uparrow}_{-}}(t)$, or else (ii) $V_{out}^{T^{\uparrow \uparrow}_{+}}(t)$ itself (which is started at time $\Delta$). Note that this reflects the fact that already the first rising input (happening at time 0) alone causes the output to eventually go to 0. Since all these trajectories only involve a single exponential function, they are easy to invert: It is apparent from \cref{outsig1} that case (i) occurs for values $\Delta \geq -\log(0.5)CR_{n_A}$, whereas \cref{outsig2} governs case (ii) for smaller values of $\Delta$.
\end{proof}

\subsubsection{Falling input transitions}
\label{sec:trajfalling}

In order to compute the MIS gate delay $\delta_{M,+}^{\uparrow}(\Delta)$,
we need to study the time the voltage trajectory
$V_{out}^{T^{\downarrow \downarrow}_{+}}(t)$ given in \cref{SoughtOutput} needs to
hit the threshold voltage $\vdd/2$ when starting from
$V_{out}^{T^{\downarrow}_{-}}(\Delta)=0$. After all, a \NOR\ gate, where both
inputs were initialized to $\vdd$ at time $-\infty$, and where only
one input experiences a falling transition at time 0 keeps its output
at 0. Consequently, at time $t=\Delta$, when the second falling input
transition occurs, the output voltage $V_{out}^{T^{\downarrow}_{-}}(\Delta)$
is still $0$.
 
Therefore, $t=\delta_{M,+}^{\uparrow}(\Delta)$ must be a solution 
of the functional equation
\begin{align}
I(t,\Delta)=e^{\frac{-t}{2RC}} \Bigl(1+\frac{2t}{d+\sqrt{\chi}}\Bigr)^{\frac{-A+a}{2RC}} \Bigl(1+\frac{2t}{d-\sqrt{\chi}}\Bigr)^{\frac{A}{2RC}} -\frac{1}{2}=0.
\label{Soughtfunction}
\end{align}

In view of the complicated shape of \cref{Soughtfunction}, it is
immediately apparent that there is not much hope to obtain an explicit
solution $\delta_{M,+}^{\uparrow}(\Delta)$ satisfying $I\bigl(\delta_{M,+}^{\uparrow}(\Delta),\Delta\bigr)=0$, for every $\Delta$. 
Even worse, since $\lim_{\Delta\to 0} A =0$ (recall \cref{eq:A}) and also
$\lim_{\Delta\to 0} (d-\sqrt{\chi}) =0$ (recall \cref{eq:d} and \cref{eq:chi}),
it is apparent that we cannot even determine a local solution of 
\cref{Soughtfunction} in a neighborhood of $(0,0)$ via the implicit
function theorem, as $(0,0)$ is a singular point (a cusp, as
already suggested by \cref{corFig5}). Fortunately, however, the 
bootstrapping method from asymptotic analysis \cite{deB70} eventually
allowed us to develop accurate asymptotic expansions, in particular, 
for $\Delta \to 0$.

In a nutshell, bootstrapping (sometimes) allows to improve the accuracy 
of an a priori known  asymptotic expansion of the sought solution
of $I(t,\Delta)=0$, by rewriting $I(t,\Delta)=0$ into a suitable equivalent 
form $t = J(t,\Delta)$, and plugging the known expansion into
the right-hand side only. In particular, relying on the fact that
$\delta_{M,+}^{\uparrow}(\Delta)=\delta_0=O(1)$ for $\Delta \to 0$
as established in \cref{lem:expansionsinfty} below, one can easily 
derive the more accurate expansion $\delta_{M,+}^{\uparrow}(\Delta)=
\delta_0 + O(\Delta)$ for $\Delta \to 0$, where $\delta_0=\delta_{M,+}^{\uparrow}(0)$ is independent of $\Delta$.

The following two technical lemmas provides asymptotic expansions
of the basic ingredients in \cref{Soughtfunction}:

\begin{lem}\label{lem:expansions} For $\Delta \to 0$, we have the
following asymptotic expansions:
\begin{flalign}
\sqrt{\chi}=& (a+ \Delta) \sqrt{1- \frac{2 \alpha_2 \Delta}{R(a+ \Delta)^2}} = \frac{\alpha_1+ \alpha_2}{2R}+ \frac{\alpha_1- \alpha_2}{\alpha_1+ \alpha_2} \Delta + O(\Delta^2), 
\label{Eq:Boot1}\\
 d+ \sqrt{\chi}&= a +\Delta+ \sqrt{\chi} = \frac{\alpha_1+ \alpha_2}{R}+ \frac{2 \alpha_1}{\alpha_1+ \alpha_2} \Delta + O(\Delta^2), 
\label{Eq:Boot2}\\
 d- \sqrt{\chi}&= a +\Delta- \sqrt{\chi} = \frac{2 \alpha_2}{\alpha_1+ \alpha_2} \Delta + O(\Delta^2), 
\label{Eq:Boot3}\\
\frac{A}{2RC} &= \frac{1}{2RC} \cdot \frac{-as_1-c'}{s_2-s_1}= 
\frac{-1}{2RC} \cdot \frac{a\frac{d-\sqrt{\chi}}{2}-\frac{\alpha_2\Delta}{2R}}{\sqrt{\chi}} = O(\Delta^2). 
\label{Eq:Boot4}
\end{flalign}
\end{lem}
\begin{proof}
For \cref{Eq:Boot1}, recalling definition \cref{eq:chi} of $\chi$ and 
using the well-known expansions $\sqrt{1+x}=1 + x/2 + O(x^2)$ and $1/(1+x)^2=1-2x +O(x^2)$ for $x \to 0$, we obtain 
\begin{equation}
\sqrt{1- \frac{2 \alpha_2 \Delta}{R(a+ \Delta)^2}} = 1 - \frac{2\alpha_2 \Delta}{2a^2R(1+\Delta/a)^2} + O(\Delta^2)
= 1 - \frac{\alpha_2 \Delta}{a^2R} + O(\Delta^2)\nonumber.
\end{equation}
Plugging this into the first equality in \cref{Eq:Boot1} and recalling
$a= \frac{\alpha_1+ \alpha_2}{2R}$, the claimed asymptotic expansion
follows by simple algebra.

\cref{Eq:Boot2}, \cref{Eq:Boot3} and \cref{Eq:Boot4} follow easily 
from their definitions \cref{eq:d} and \cref{eq:A} by plugging in
the asymptotic expansion of $\sqrt{\chi}$ given in \cref{Eq:Boot1}.
\end{proof}

\begin{lem}\label{lem:expansionsinfty} For $\Delta \to \infty$, we have the
following asymptotic expansions:
\begin{flalign}
\sqrt{\chi}=& (a+ \Delta) \sqrt{1- \frac{2 \alpha_2 \Delta}{R(a+ \Delta)^2}} = \Delta + \frac{\alpha_1-\alpha_2}{2R}+ O(\Delta^{-1}), 
\label{Eq:Boot1infty}\\
 d+ \sqrt{\chi}&= a +\Delta+ \sqrt{\chi} = 2\Delta + \frac{\alpha_1}{R} + 
O(\Delta^{-1}), 
\label{Eq:Boot2infty}\\
 d- \sqrt{\chi}&= a +\Delta- \sqrt{\chi} = \frac{\alpha_2}{R} + O(\Delta^{-1}), 
\label{Eq:Boot3infty}\\
\frac{A}{2RC} &= \frac{1}{2RC} \cdot \frac{-as_1-c'}{s_2-s_1}= \frac{-1}{2RC} \cdot \frac{a\frac{d-\sqrt{\chi}}{2}-\frac{\alpha_2\Delta}{2R}}{\sqrt{\chi}} =
\frac{\alpha_2}{4R^2C} + O(\Delta^{-1}). 
\label{Eq:Boot4infty}
\end{flalign}
\end{lem}
\begin{proof}
For \cref{Eq:Boot1infty}, recalling definition \cref{eq:chi} of $\chi$ and again
using the well-known expansions $\sqrt{1+x}=1 + x/2 + O(x^2)$ and $1/(1+x)^2=1-2x +O(x^2)$ for $x \to 0$, we obtain 
\begin{flalign}
\sqrt{1- \frac{2 \alpha_2 \Delta}{R(a+ \Delta)^2}} &= \sqrt{1- \frac{2 \alpha_2}{\Delta R(1+ a/\Delta)^2}}\nonumber
= 1 - \frac{\alpha_2}{\Delta R(1+a/\Delta)^2} + O(\Delta^{-2})\nonumber\\
&= 1 - \frac{\alpha_2}{\Delta R} + O(\Delta^{-2})\nonumber.
\end{flalign}
Plugging this into the first equality in \cref{Eq:Boot1infty} and recalling
$a= \frac{\alpha_1+ \alpha_2}{2R}$, the claimed asymptotic expansion
follows by simple algebra.

\cref{Eq:Boot2infty}, \cref{Eq:Boot3infty} and \cref{Eq:Boot4infty} follow easily 
from their definitions \cref{eq:d} and \cref{eq:A} by plugging in
the asymptotic expansion of $\sqrt{\chi}$ given in \cref{Eq:Boot1infty}.
\end{proof}

As the basis for our first bootstrapping step, we will need the extremal
delay values $\delta_{M,+}^{\uparrow}(0)$, $\delta_{M,+}^{\uparrow}(\infty)$ 
and $\delta_{M,+}^{\uparrow}(-\infty)$. The following \cref{lem:extremaldelays} 
will provide solutions $\delta_0$, $\delta_\infty$ and $\delta_{-\infty}$
of \cref{Soughtfunction} for $\Delta=0$, $\Delta=\infty$ and $\Delta=-\infty$,
respectively, which can be expressed in terms of some branch of the 
multi-valued \emph{Lambert $W$ function} \cite{CGHJ96}. We note that
$W(x)$ provides the inverse of the function $ye^y=x$, and has only two
real-valued branches: the principal branch $y=W_0(x)$ where
$y \geq -1$, and the branch $y=W_{-1}(x)$ where $y \leq -1$. Since we
will also prove in \cref{thm:expdelayfunctions} later on that \cref{Soughtfunction} 
has a unique solution for $\Delta=0$, $\Delta=\infty$ and $\Delta=-\infty$,
it follows that indeed $\delta_{M,+}^{\uparrow}(0)=\delta_0$, $\delta_{M,+}^{\uparrow}(\infty)=\delta_{\infty}$ and $\delta_{M,+}^{\uparrow}(-\infty)=\delta_{-\infty}$.

\begin{lem}[Extremal MIS delay values]\label{lem:extremaldelays}
Given\footnote{In \cref{Sec:Param}, we 
will explain how to determine these parameters from given delay values
$\delta_{0}$, $\delta_{\infty}$ and $\delta_{-\infty}$.} $\alpha_1$, $\alpha_2$, $R$, and $C$, we find
\begin{align}
\delta_{0} &= - \frac{\alpha_1 + \alpha_2}{2R} \Bigl[ 1+ W_{-1}\Bigl(\frac{-1}{e \cdot 2^{\frac{4R^2C}{\alpha_1+ \alpha_2}}}\Bigr) \Bigr],  \label{eq:delta0} \\
\delta_{\infty}&= -\frac{\alpha_2}{2R} \Bigl[ 1+ W_{-1}\Bigl(\frac{-1}{e \cdot 2^{\frac{4R^2C}{\alpha_2}}}\Bigr) \Bigr],  \label{eq:deltainf} \\
\delta_{-\infty}&= -\frac{\alpha_1}{2R} \Bigl[ 1+ W_{-1}\Bigl(\frac{-1}{e \cdot 2^{\frac{4R^2C}{\alpha_1}}}\Bigr) \Bigr]. \label{eq:deltaminf}
\end{align}
\end{lem}
\begin{proof}
We start with the proof for $\delta_0$. 
Plugging in $\Delta=0$ in \cref{Soughtfunction} leads to
\begin{equation}
e^{-\frac{\delta_0}{2RC}} \Bigl(1+\frac{\delta_0}{a}\Bigr)^{\frac{a}{2RC}}=\frac{1}{2}, \label{eq:traj0}
\end{equation}
since $d+\sqrt{\chi}=2a$, $d-\sqrt{\chi}=0$ by \cref{eq:d}--\cref{eq:chi}, and
$A=0$ by \cref{eq:A}; note that the third factor in \cref{Soughtfunction} 
collapses to 1 since $(1+\infty)^0 = 1$. Raising \cref{eq:traj0} to the power 
$2RC/a$, one obtains
\begin{equation}
e^{-\frac{\delta_0}{a}}\Bigl(1+\frac{\delta_0}{a}\Bigr)= 2^{-\frac{2RC}{a}}. \label{eq:traj0raised}
\end{equation}
Setting $y = - (1+ \frac{\delta_0}{a})$ and $\gamma = \frac{2^{-\frac{2RC}{a}}}{e}$, this translates to $e^y y = - \gamma$. Note carefully that $-\gamma > -\frac{1}{e}$ and  $y<-1$. It hence follows that 
\begin{equation}
y= W_{-1} \Bigl( \frac{-1}{e \cdot 2^{\frac{2RC}{a}}} \Bigr), \nonumber
\end{equation}
which is equivalent to \cref{eq:delta0} by recalling $y = - (1+ \frac{\delta_0}{a})$ and $a= \frac{\alpha_1+ \alpha_2}{2R}$.

We next turn our attention to $\delta_{\infty}$. 
Recalling the asymptotic expansions in \cref{lem:expansionsinfty}, it is
not difficult to verify that plugging in $\Delta=\infty$ in \cref{Soughtfunction} leads to
\begin{equation}
e^{-\frac{\delta_{\infty}}{2RC}} \Bigl(1+\frac{\delta_{\infty}}{\frac{\alpha_2}{2R}}\Bigr)^{\frac{\alpha_2}{4R^2C}}=\frac{1}{2}; \label{eq:traj0infty}
\end{equation}
note that it is the second factor in \cref{Soughtfunction} that
collapses to 1 here. Since \cref{eq:traj0infty} differs from \cref{eq:traj0}
only in that $a=\frac{\alpha_1+\alpha_2}{2R}$ has been replaced by  
$\frac{\alpha_2}{2R}$, the above derivations can be literally used to
also confirm \cref{eq:deltainf}, and, by our usual symmetry argument,
\cref{eq:deltaminf}.
\end{proof}

We are now ready for proving the main \cref{thm:expdelayfunctions} of this 
section:

\begin{theorem}[MIS Delay functions for falling input transitions]\label{thm:expdelayfunctions}
For any $0 \leq |\Delta| \leq \infty$, the MIS delay functions of our model for falling input transitions are given by
\begin{flalign}
\delta_{M,+}^{\uparrow}(\Delta) &= \begin {cases}
\delta_{0} - \frac{\alpha_1}{\alpha_1+\alpha_2} \Delta  &   \ \ 0 \leq \Delta < \frac{(\alpha_1+\alpha_2)(\delta_{0} - \delta_{\infty})}{\alpha_1}   \\ 
\delta_{\infty} &   \ \ \Delta \geq \frac{(\alpha_1+\alpha_2)(\delta_{0} - \delta_{\infty})}{\alpha_1}
\end {cases}\label{Risingdelayformula}\\
\delta_{M,-}^{\uparrow}(\Delta) &= \begin {cases}
\delta_{0} - \frac{\alpha_2}{\alpha_1+\alpha_2} |\Delta|  &   \ \ 0 \leq |\Delta| < \frac{(\alpha_1+\alpha_2)(\delta_{0} - \delta_{-\infty})}{\alpha_2}   \\ 
\delta_{-\infty} &   \ \ |\Delta| \geq \frac{(\alpha_1+\alpha_2)(\delta_{0} - \delta_{-\infty})}{\alpha_2}
\end {cases}\label{Risingdelayformulaminus}
\end{flalign}
\end{theorem}

\begin{proof}
Since inverting \cref{Soughtfunction} globally is hopeless, we will determine
the linear asymptotic expansion of $\delta_{M,+}^{\uparrow}(\Delta)$ for 
$\Delta \to 0$ and the constant asymptotic expansions of $\delta_{M,+}^{\uparrow}(\Delta)$ for $\Delta \to \pm\infty$, and glue them together at their intersection
point.

To simplify our derivations, we will employ the variable substitutions
\begin{equation}
y=\frac{t}{2RC}   \quad\mbox{and}\quad  x=\frac{\Delta}{2RC},\nonumber
\end{equation}
in \cref{Soughtfunction}. This leads to
\begin{flalign}
 e^{-y} \Bigl( 1+\frac{y}{p+p_1x+O(x^2)}\Bigr)^{p+O(x^2)} \Bigl( 1+\frac{y}{mx+O(x^2)} \Bigr)^{O(x^2)} =\frac{1}{2}, &
\label{Soughtfunction2}
\end{flalign}
where the constants
\begin{equation}
p= \frac{\alpha_1+\alpha_2}{4R^2C},\quad p_1= \frac{\alpha_1}{\alpha_1+\alpha_2}, \quad m=\frac{\alpha_2}{\alpha_1+\alpha_2} \nonumber
\end{equation}
follow from the expansions provided in \cref{lem:expansions}.

Obviously, in accordance with \cref{eq:traj0}, setting $x=0$ in the logarithm of \cref{Soughtfunction2} results in the following equation for the solution(s) $y_0$:
\begin{equation}
y_0=p \cdot \log(1+\frac{y_0}{p})+\log(2).
\label{eq:y_0}
\end{equation}
Using continuity and convexity arguments, we first show that \cref{eq:y_0} 
and hence \cref{Soughtfunction2}
has a unique solution, which must hence be equal to $y_0=\delta_0/(2RC) >0$ according to \cref{lem:extremaldelays}. More specifically, given $y_1$ and $y_2$ with $0< y_1 < y_0 < y_2$, we prove that \cref{Soughtfunction2} has a unique solution $y = y(x) \in [y_1,y_2]$ for sufficiently small $x$ by using Banach's fixed point theorem \cite{rudin1976principles}:
By taking the logarithm, we can rewrite \cref{Soughtfunction2}
as a fixed point equation
\begin{equation}
\label{eqfixedpoint}
y = \bigl(p+O(x^2)\bigr)\log\Bigl( 1+\frac{y}{p+p_1x+O(x^2)}\Bigr) + O(x^2) \log \Bigl( 1+\frac{y}{mx+O(x^2)} \Bigr) - \log \frac 12.
\end{equation}
Since
\begin{flalign}
\frac{ \partial}{\partial y}\bigl(p+O(x^2)\bigr)\log\Bigl( 1+\frac{y}{p+p_1x+O(x^2)}\Bigr) &= 
\frac{\frac{p+O(x^2)}{p+p_1x+O(x^2)}}{1+  \frac{y}{p+p_1x+O(x^2)}} \nonumber \\
& \le  \left(1+  \frac{y_1}{2p}\right)^{-1} < 1 \nonumber
\end{flalign}
provided $x$ is chosen sufficiently small, and
\[
\frac{ \partial}{\partial y} 
\left( O(x^2) \log \Bigl( 1+\frac{y}{mx+O(x^2)} \Bigr) \right) 
= O\left( \frac{x^2}{ mx+O(x^2) + y } \right) = O(x^2),
\]
it follows that \cref{eqfixedpoint} is a contraction. 
Banach's fixed point theorem thus shows that 
the solution $y(x)$ (and hence also the corresponding solution
$t(\Delta)$ of \cref{Soughtfunction}) is unique. Note that an analogous reasoning 
can be used to prove that the solutions for $\Delta\to \infty$ and
$\Delta \to -\infty$ are unique, which also confirms the values $\delta_\infty$
and $\delta_{-\infty}$ given in \cref{lem:extremaldelays}.

For our bootstrapping step, we write
\begin{equation}
y=y_0 + z
\label{eq:linboots}
\end{equation}
with $z=z(x)\in [y_1-y_0,y_2-y_0]$ for $x \to 0$, 
and show next that actually $z=O(x)$. We again take the logarithm of \cref{Soughtfunction2}
and split it up into three parts $T_1, T_2, T_3$. Furthermore, we set
\[
U = \frac{z}{p + p_1x + O(x^2)} - \frac{y_0 p_1}{p^2} x + O(x^2),
\]
which can be made arbitrarily small by choosing $y_1, y_2$ appropriately,
and obtain 
\begin{equation}
z = U(p+p_1x) + \frac{y_0 p_1}{p} x + O(x^2).\label{eq:zU}
\end{equation}
Using the relations $1/(1+x)=1-x + O(x^2)$ and $\log(1+x)=x + O(x^2)$ for $x\to 0$,
we thus obtain
\begin{flalign}
T_1 &= \log\bigl(e^{-y}\bigr)= - y_0-z = -y_0 - U(p+p_1x) - \frac{y_0 p_1}{p} x + O(x^2),  \nonumber
\end{flalign}
\begin{flalign}
T_2 &= \log\Bigl(1+ \frac{y}{p+p_1x+O(x^2)}\Bigr)^{p+O(x^2)}  \nonumber \\
&= \bigl(p+O(x^2)\bigr) \log\Bigl(1+ \frac{y_0}{p+p_1x+O(x^2)} + \frac{z}{p+p_1x+O(x^2)}\Bigr)  \nonumber \\
&= \bigl(p+O(x^2)\bigr) \log\Bigl(1+ \frac{y_0}{p}\cdot\frac{1}{1+\frac{p_1x}{p}+O(x^2)} + U + \frac{y_0p_1x}{p^2} + O(x^2) \Bigr)  \nonumber \\
&= \bigl(p+O(x^2)\bigr) \log\Bigl(1+ \frac{y_0}{p} + O(x^2) + U\Bigr) = \bigl(p+O(x^2)\bigr) \log\Bigl(\bigl(1+ \frac{y_0}{p}\bigr)\bigl(1 + \frac U{1+\frac{y_0}{p}}\bigr)\Bigr) \nonumber \\
&= \bigl(p+O(x^2)\bigr) \log\bigl(1+\frac{y_0}{p}\bigr) + \bigl(p+O(x^2)\bigr) \frac {U+O(U^2)}{1+\frac{y_0}{p}} \nonumber
\end{flalign}
\begin{flalign}
T_3 &= \log \Bigl(1+\frac{y}{mx+O(x^2)} \Bigr)^{O(x^2)}  = \log\Bigl(\frac{y+mx+O(x^2)}{mx+O(x^2)} \Bigr)^{O(x^2)} \nonumber \\
&= O(x^2) \cdot  \log \bigl(y_0 +z+mx+O(x^2)\bigr) +O\bigl(x^2 \log(x)\bigr) \nonumber \\
&= O(x^2) \cdot \log \left(y_0\Bigl(1+\frac{z+mx+O(x^2)}{y_0}\Bigr)  \right) +O\bigl(x^2 \log(x)\bigr) \nonumber \\
&= O(x^2) \cdot \Bigl( \log(y_0) +\frac{z+mx+O(x^2)}{y_0} \Bigr) +O\bigl(x^2 \log(x)\bigr) \nonumber\\
& =  O\bigl(x^2 \log(x)\bigr). \nonumber
\end{flalign}
Thus, \cref{Soughtfunction2}, which is equivalent to $T_1+T_2+T_3=\log(1/2)$,
is also equivalent to 
\begin{equation}
\label{eqU}
- U\bigl(p+p_1x\bigr) - \frac{y_0 p_1}{p} x 
+ \bigl(p+O(x^2)\bigr) \frac {U\bigl(1+O(U)\bigr)}{1+\frac{y_0}{p}} + O(x^2\log x)  = 0,
\end{equation}
where the terms involving $y_0$ canceled out due to relation \cref{eq:y_0}. Extracting $U$ and recalling
that $O(U)$ can be made arbitrarily small by choosing $y_1,y_2$ appropriately finally reveals 
$U=O(x)$ and hence $z=z(x) = O(x)$ by \cref{eq:zU} as claimed.

In an additional bootstrapping step, we can be slightly more precise w.r.t.\ $T_2$ and obtain
\begin{flalign}
T_2&=\bigl(p+O(x^2) \bigr) \cdot \log\Bigl(1+ \frac{y}{p+p_1x+O(x^2)}\Bigr)  \nonumber\\
& =\bigl(p+O(x^2) \bigr) \cdot \log \Bigl(1+\frac{y_0+z}{p}\bigl(1-\frac{p_1}{p}x + O(x^2)\bigr) \Bigr)  \nonumber \\
& = \bigl(p+O(x^2) \bigr) \cdot \log \left( \bigl(1+\frac{y_0}{p}\bigr)\Bigl(1+\frac{\frac{z}{p}- \frac{y_0 p_1}{p^2}x +O(x^2)}{1+\frac{y_0}{p}}\Bigr) \right) \nonumber \\
& =\bigl(p+O(x^2) \bigr) \cdot \Bigl( \log\bigl(1+\frac{y_0}{p}\bigr)+ \frac{z}{y_0+p} - \frac{y_0p_1}{p(y_0+p)}x + O(x^2)  \Bigr) \nonumber \\
& = \frac{p}{y_0+p}z + p \log\bigl(1+\frac{y_0}{p}\bigr)- \frac{y_0p_1}{y_0+p}x + O(x^2).  \nonumber
\end{flalign}

This leads to 
\begin{align}
-y_0-z +\log(2) + \frac{p}{y_0+p}z  + p \log(1+ \frac{y_0}{p})  - \frac{y_0 p_1}{y_0+p}x   + O\bigl(x^2 \log(x)\bigr) =0,
\end{align}
which, by virtue of \cref{eq:y_0}, gives
\begin{align}
 z=-p_1 x + O\bigl(x^2 \log(x)\bigr) = -\frac{\alpha_1}{\alpha_1+\alpha_2}x + O\bigl(x^2 \log(x)\bigr).
\end{align}
Recalling \cref{eq:linboots}, we therefore arrive at the improved expansion
\begin{align}
y=y_0-\frac{\alpha_1}{\alpha_1+\alpha_2}x +O\bigl(x^2 \log(x)\bigr)
\label{eq:linboots2}
\end{align}
and, after undoing our variable substitution,
\begin{align}
\delta_{M,+}^{\uparrow}(\Delta)=\delta_0 - \frac{\alpha_1}{\alpha_1+\alpha_2} \Delta + O(\Delta^2 \log(\Delta)).
\label{zero-exp}
\end{align}

Finally, it is easy to check that the crossing point of the linear part of
\cref{zero-exp} and $\delta_{\infty}$ is $\Delta = \frac{(\alpha_1+\alpha_2)(\delta_{0} - \delta_{\infty})}{\alpha_1}$. By pasting them together at this crossing point, we obtain the delay formula \cref {Risingdelayformula} that is valid for all values of $\Delta$.

Last but not least, $\delta_{M,-}^{\uparrow}(\Delta)$ is obtained by exchanging $\alpha_1$ and $\alpha_2$ and replacing $\delta_{\infty}$ by $\delta_{-\infty}$ as well
as $\Delta$ by $|\Delta|$ in \cref{Risingdelayformula}, according
to our usual symmetry argument, which completes our proof.
\end{proof}

To conclude this section, we note that more accurate asymptotic expansions
for the delay can be derived easily by further bootstrapping steps. It turns out, however,
that improving the accuracy for $\Delta$ very close to 0 has its price
in a rapid worsening of the accuracy for larger values of $\Delta$. 
Consequently, just pasting together the expansions for $\Delta \to 0$ 
and $\Delta \to \pm\infty$ would no longer be sufficient to cover 
the whole range for $\Delta$. Whereas bootstrapping could
also be used to develop an asymptotic expansion at some intermediate
point within this gap, the resulting improvement is not worth 
the effort.

\subsection{Model parametrization and evaluation results}
\label{Sec:Param}
What is still needed to use our model, in particular, the delay formulas established in \cref{thm:delayfunctions} and \cref{thm:expdelayfunctions}, is a practical procedure for model parametrization: Given some
data that characterize the delays of a real gate, one needs to determine
appropriate values for the model parameters $\alpha_1$, $\alpha_2$, $C$, $R$, $R_{n_A}$, and $R_{n_B}$ and an appropriate pure delay $\dmin$ that 
align our model with these data.

As in \cite{FMOS22:DATE, ferdowsi2023accurate}, we will parameterize our
model based on the characteristic MIS delay values $\ddoD_S(-\infty)$, $\ddoD_S(0)$, and $\ddoD_S(\infty)$ according to \cref{corFig3} and 
 $\dupD_S(-\infty)$, $\dupD_S(0)$, and $\dupD_S(\infty)$ according to 
\cref{corFig5}. In
sharp contrast to the parametrization procedure employed for the original model in 
\cite{ferdowsi2023accurate}, which was based on least-squares fitting, we
can exploit our explicit trajectory formulas to get
rid of any fitting. In fact, as already in \cref{lem:extremaldelays}, Lambert $W$
functions will turn out to be instrumental also here.

\begin{theorem}[Gate characterization]\label{thm:gatechar}
Let $\ddoD_S(-\infty)$,  $\ddoD_S(0)$, $\ddoD_S(\infty)$ and $\dupD_S(-\infty)$,  $\dupD_S(0)$, $\dupD_S(\infty)$ be the MIS delay values of a real gate that
shall be matched by our model, in the sense that $\ddoD_{M,-}(-\infty)=\ddoD_S(-\infty)$, $\ddoD_{M,-}(0)=\ddoD_{M,+}(0)=\ddoD_S(0)$, $\ddoD_{M,+}(\infty)=\ddoD_S(\infty)$ and $\dupD_{M,-}(-\infty)=\dupD_S(-\infty)$,  $\dupD_{M,-}(0)=\dupD_{M,+}(0)=\dupD_S(0)$, $\dupD_{M,+}(\infty)=\dupD_S(\infty)$.

Given an arbitrily chosen value $C$ for the load capacitance, this matching is accomplished by choosing the model parameters as follows:
\begin{flalign}
\dmin &= \ddoD_S(0) - \sqrt{\bigl(\ddoD_S(\infty)-\ddoD_S(0)\bigr)\bigl(\ddoD_S(-\infty)-\ddoD_S(0)\bigr)} \label{eq:dmin}\\
R_{n_{B}}&=\frac{\ddoD_S(-\infty)-\dmin}{C \cdot \log(2)} \label{eq:rnb}  \\
R_{n_{A}}&= \frac{\ddoD_S(\infty)-\dmin}{C \cdot \log(2)} \label{eq:rna}
\end{flalign}
Furthermore, using the function
\begin{equation}
A(t,R,C)=\frac{-2R \bigl(t-2RC \cdot \log(2) \bigr)}{W_{-1}\Bigl( \bigl(\frac{2RC \cdot \log(2)}{t}-1\bigr) e^{\frac{2RC \cdot \log(2)}{t}-1} \Bigl) +1 - \frac{2RC \cdot \log(2)}{t}}\label{eq:AtRC},
\end{equation}
determine $R$ by numerically\footnote{Whereas there might be a way to solve it analytically, we did not find it so far.} solving the equation 
\begin{equation}
A\bigl(\dupD_S(0)-\dmin,R,C\bigr)-A\bigl(\dupD_S(\infty)-\dmin,R,C\bigr) - A\bigl(\dupD_S(-\infty)-\dmin,R,C\bigr) = 0 \label{eq:forR},
\end{equation}
and finally choose
\begin{flalign}
\alpha_1 &= A\bigl(\dupD_S(-\infty)-\dmin,R,C\bigr)\label{eq:alpha1},\\
\alpha_2 &= A\bigl(\dupD_S(\infty)-\dmin,R,C\bigr)\label{eq:alpha2}.
\end{flalign}
\end{theorem}
\begin{proof}
We first consider the parameters determined by the rising input transition 
case. 
To align the 
delay formulas in \cref{thm:delayfunctions} with the given 
delay values, we just plug in $\ddoD_S(-\infty)-\dmin$, $\ddoD_S(0)-\dmin$, and $\ddoD_S(\infty)-\dmin$ in order to obtain the following system of equations for our
sought parameters $\dmin$, $R_{n_{B}}$ and $R_{n_{A}}$:
\begin{flalign}
&\ddoD_S(0)-\dmin- \frac{\log(2) \cdot C  \cdot  R_{n_{A}} R_{n_{B}}}{R_{n_{A}}+R_{n_{B}}}=0 \nonumber \\
& \ddoD_S(\infty)-\dmin- \log(2) \cdot C  \cdot R_{n_{A}}=0 \nonumber \\
& \ddoD_S(-\infty)-\dmin- \log(2) \cdot C  \cdot R_{n_{B}}=0 \nonumber
\end{flalign}
Some straightforward algebra shows that this system is equivalent to the following
one:
\begin{flalign}
\frac{1}{R_{n_{A}}}+\frac{1}{R_{n_{B}}} &= \frac{\log(2) \cdot C}{\ddoD_S(0)-\dmin}
\nonumber \\
\frac{1}{R_{n_{A}}} &= \frac{\log(2) \cdot C}{\ddoD_S(\infty)-\dmin}
\nonumber \\
\frac{1}{R_{n_{B}}} &= \frac{\log(2) \cdot C}{\ddoD_S(-\infty)-\dmin}
\nonumber
\end{flalign}
It follows that
\begin{equation}
\frac{1}{\ddoD_S(0)-\dmin} = \frac{1}{\ddoD_S(\infty)-\dmin} + \frac{1}{\ddoD_S(-\infty)-\dmin}, \nonumber
\end{equation}
which can be rewritten into a quadratic equation for $\dmin$, namely,
\begin{equation}
\dmin^2 - 2\ddoD_S(0)\dmin + \ddoD_S(0)\ddoD_S(\infty) + \ddoD_S(0)\ddoD_S(-\infty) - \ddoD_S(\infty)\ddoD_S(-\infty) = 0\nonumber.
\end{equation}
It is easy to verify that it has the solution stated in \cref{eq:dmin}.
Note that we need to take the negative solution in order to ensure that
$\dmin \leq \ddoD_S(0)$.

We now turn our attention to the parameters determined by the
falling input transition case. We first justify \cref{eq:AtRC} by considering 
$A\bigl(\dupD_S(0)-\dmin,R,C\bigr)$, which corresponds to setting
$t=\delta_0=\dupD_S(0)-\dmin$ as defined in \cref{lem:extremaldelays}.
Abbreviating $\alpha = \alpha_1+\alpha_2$ and noting that $\alpha=2Ra$
according to \cref{eq:a}, we start from \cref{eq:traj0raised} in the proof
of \cref{lem:extremaldelays}, which states
$e^{\frac{-2R\delta_0}{\alpha}} (1+ \frac{2R\delta_0}{\alpha})=2^{\frac{-4R^2C}{\alpha}}$. By raising  both sides to the power of $\alpha/(2R)$, we get $1 < (1+ \frac{2R\delta_0}{\alpha})^{\frac{\alpha}{2R}}= 2^{-2RC}e^{\delta_0}$ that is equivalent to $(1+ \frac{\omega}{y})^y=\beta$ with $\omega = 2R\delta_0 > 0$, $y = \alpha > 0$, and $\beta = e^{2R(\delta_0-2RC\log(2))} > 1$. Once again, by the substitution $z = 1+ \frac{\omega}{y}> 1$, we get $e^{\frac{\omega}{z-1}\log(z)}= \beta$. Accordingly, by taking the natural logarithm on both sides, we arrive at
\begin{align}
\label{eq:Thm4term1}
&\log(z)=(z-1) \gamma,
\end{align}
for $\gamma = \frac{\log(\beta)}{\omega}> 0$. We need to solve \cref{eq:Thm4term1} for $z>1$ so as to obtain $\alpha=y=\frac{\omega}{z}$. 
From \cref{eq:Thm4term1}, we get by exponentiation $z e^{-z \gamma}= e^{-\gamma}$, 
and multiplication by $-\gamma$ finally gives us $-z \gamma e^{- z \gamma}= - \gamma e^{-\gamma}$. We can solve this equation for $-z\gamma$ by means of the Lambert $W$
function. Since $\gamma > 0$ and we need the solution to satisfy $z > 1$, we
must take the branch $W_{-1}$ here to compute
\begin{align}
&z= - \frac{W_{-1}(-\gamma e^{-\gamma})}{\gamma}. \nonumber
\end{align}
Plugging in the values of $z$ and $\gamma$ into $y= \frac{\omega}{z}$, we obtain
\begin{align}
\label{eq:Thm4term2}
&y= - \frac{-\log(\beta)}{W_{-1}(-\frac{\log(\beta)}{\omega} \beta^{\frac{-1}{\omega}})+ \frac{\log(\beta)}{\omega}}.
\end{align}
Finally, replacing $\omega$ resp.\ $\beta$ by their ``generic'' value 
$\omega = 2Rt$ resp.\ $e^{2R(t-2RC\log(2))}$ (where
$\delta_0$ is replaced by $t$) in \cref{eq:Thm4term2} gives \cref{eq:AtRC}.

It only remains to justify \cref{eq:alpha2} and \cref{eq:alpha1}, where
we use the same argument as in the proof of \cref{lem:extremaldelays}:
We can use literally the same derivations as above, except that we
start from the variant of \cref{eq:traj0raised} where $a$ is replaced by
$\frac{\alpha_2}{2R}$ resp.\ $\frac{\alpha_1}{2R}$ for 
\cref{eq:alpha2} resp.\ \cref{eq:alpha1}. This finally also explains
why we can determine $R$ by (numerically) solving \cref{eq:forR}.
\end{proof}

We are now ready to compare the delay predictions of our 
\cref{thm:delayfunctions} and \cref{thm:expdelayfunctions}
with the ones provided in the original paper \cite{ferdowsi2023accurate}.
For that purpose, we employ our \cref{thm:gatechar} for computing the parameters
for the same $15$nm technology CMOS \NOR\ gate used in \cite{ferdowsi2023accurate}, which are summarized in \cref{table:params15}, and visualize the delay predictions of our model: \cref{fig:CharlieResult} depicts our delays (blue curve) and compares it to the analog reality (dashed red curve), as well as to the predictions
provided by the original model in \cite{ferdowsi2023accurate} (dashed orange curve).

\begin{table}[h]
\centering
\caption{Model parameter values for the $15$nm CMOS \NOR\ gate used for producing \cref{corFig3} and \cref{corFig5}. The chosen capacitance value is $C=3.6331599443276\; fF$.}
\label{table:params15}
\scalebox{0.8}{
\begin{tabular}{cccc}
\hline
\multicolumn{4}{|c|}{Parameters determined by falling output transitions} \\
\multicolumn{1}{|c|}{$\dmin= 16.963423585525\; ps$} & \multicolumn{1}{c|}{$R_{n_A}=8.760489389736\; k\Omega$} & \multicolumn{2}{c|}{$R_{n_B}=8.658111065573\; k\Omega$} \\ \hline  \hline
\multicolumn{4}{|c|}{Parameters determined by rising output transitions} \\
\multicolumn{1}{|c|}{$R=6.539995525955\; k\Omega$} & \multicolumn{1}{c|}{$\alpha_1=20.4461\cdot 10^{-9}\; \Omega s$} & \multicolumn{2}{c|}{$\alpha_2=9.3487\cdot 10^{-9}\; \Omega s$} \\ \hline
\end{tabular}}
\end{table}

\begin{figure}[h]
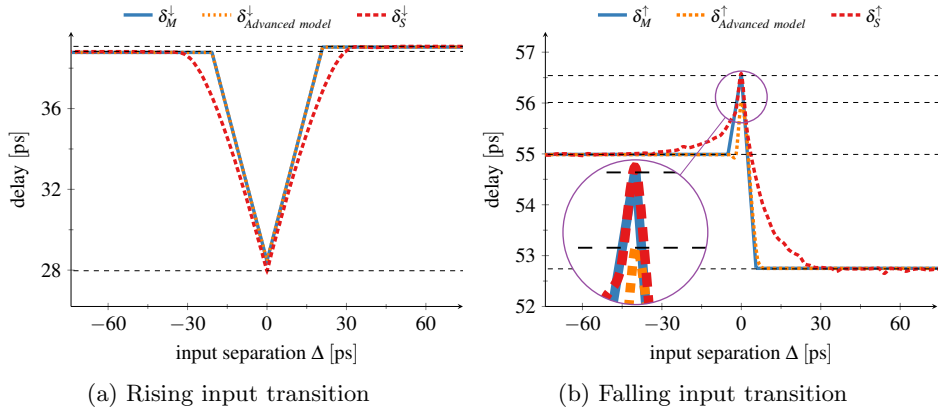

  \centering
  \subfloat[Rising input transition]{
    \includegraphics[width=0.485\linewidth]{\figPath{hm_falling_output_new_Nahs.pdf}}%
    \label{Fig.Resup}}
  \hfil
  \subfloat[Falling input transition]{
    \includegraphics[width=0.485\linewidth]{\figPath{hm_rising_output_complete_Nahs_magnify.pdf}}%
    \label{Fig.Resdown}}
  \caption{Our proposed computed ($\delta_M^{\downarrow/\uparrow}(\Delta)$), the measured ($\delta_S^{\downarrow/\uparrow}(\Delta)$), and those computed by the baseline model in \cite{ferdowsi2023accurate} for the $15$nm
CMOS \NOR\ gate from \cite{FMOS22:DATE}.}\label{fig:CharlieResult}
\end{figure}

Since we are utilizing almost the same delay formulas and parameters as those used in \cite{ferdowsi2023accurate} for the rising input transition case, we observe identical blue and orange curves, which closely match the gate's real MIS delays.  There is a significant improvement in the delay predictions of our model over the one in \cite{ferdowsi2023accurate} for the falling input transition case, however. In particular, according to \cref{Fig.Resdown}, our model accurately predicts the real delays even for very small values of $\Delta$, where \cite{ferdowsi2023accurate} is considerably off.

\section{Conclusions}	
\label{sec:conclusions}

We presented a general continuity proof for a broad class of first-order thresholded
hybrid models, which arise naturally in digital integrated circuits. We showed that, under
mild conditions regarding causality, digitized hybrid gates can be composed to form
circuits with unique and well-behaved executions. We introduced the intricacies of multi-input switching effects in multi-input gates and proved the continuity of  two state-of-the-art digitized hybrid models for CMOS \NOR\ gates. Moreover, we
revisited the currently best of these models and provided a completely new
analysis of its MIS delay predictions, based on explicit solutions of the
involved ODEs, which not only resulted in a much better
accuracy but also in an explicit model parametrization procedure.

 \bibliographystyle{elsarticle-num} 
\bibliography{mybib}

\end{document}